\documentclass{LMCS}
\usepackage{proof}
\usepackage{amsmath}
\usepackage{amssymb}
\usepackage{subfigure}
\usepackage{epsf}
\usepackage{graphics}
\usepackage{wrapfig}
\usepackage{mathrsfs}
\usepackage{enumerate,hyperref}

\newcommand{\lin}{\ensuremath{\multimap}}
\newcommand{\lam}{\ensuremath{\lambda}}
\newcommand{\vdashNEAL}{\ensuremath{\vdash_{\mathsf{NEAL}}}}
\newcommand{\ea}{{\sf EA}}
\newcommand{\contr}[4]{\ensuremath{[#1]_{#2=#3,#4}}}
\newcommand{\promote}[5]{\ensuremath{!\left(#1\right)
    \left[{}^{#2}/#3,\ldots,{}^{#4}/#5\right]}}
\newcommand{\promotewa}[1]{\ensuremath{!\left(#1\right)}}
\newcommand{\LambdaEA}{\ensuremath{\Lambda^{\ea}}}

\newcommand{\linear}{\multimap}
\newcommand{\A}{\mathbb{A}}
\newcommand{\N}{\mathbb{N}}
\newcommand{\U}{\mathbb{U}}
\newcommand{\C}{\mathbb{C}}
\newcommand{\bn}{\rightarrow_{n} }
\newcommand{\bv}{\rightarrow_{v} }

\newcommand{\PT}{\ensuremath{PT}}
\newcommand{\NEAL}{\ensuremath{\textsf{NEAL}}}
\newcommand{\ETAS}{\ensuremath{\textsf{ETAS}}}

\newcommand{\LAL}{\ensuremath{\textsf{LAL}}}
\newcommand{\DLAL}{\ensuremath{\textsf{DLAL}}}

\newcommand{\SAL}{\ensuremath{\textsf{SAL}}}
\newcommand{\EAL}{\ensuremath{\textsf{EAL}}}

\newcounter{number}

\def\doi{4 (4:5) 2008}
\lmcsheading%
{\doi}
{1--28}
{}
{}
{May~\phantom{.}14, 2007}
{Nov.~\phantom{0}7, 2008}
{}   

\begin{document}

\title[Light Logics and the Call-by-Value Lambda Calculus]{Light Logics and the Call-by-Value Lambda Calculus}

\author[P.~Coppola]{Paolo Coppola\rsuper a}
\address{{\lsuper a}Dipartimento di Matematica e Informatica, Universit\`a di Udine.
         via delle Scienze 206, 33100 Udine, Italy.}
\email{coppola@dimi.uniud.it}
\author[U.~Dal Lago]{Ugo Dal Lago\rsuper b}
\address{{\lsuper b}Dipartimento di Scienze dell'Informazione, Universit\`a di Bologna.
         via Mura Anteo Zamboni 7, 40127 Bologna, Italy}
\email{dallago@cs.unibo.it}
\author[S.~Ronchi Della Rocca]{Simona Ronchi Della Rocca\rsuper c}
\address{{\lsuper c}Dipartimento di Informatica, Universit\`a di Torino
         corso Svizzera 185, 10129 Torino, Italy}
\email{ronchi@di.unito.it}

\thanks{{\lsuper{a,b,c}}All authors have been partially supported by PRIN project ``FOLLIA'' and french ANR ``NO-CoST'' project (JC05\_43380).}


\begin{abstract}
The so-called light logics~\cite{Girard98ic,Asperti98lics,Asperti02tocl} 
have been introduced as logical systems enjoying quite remarkable
normalization properties.
Designing a type assignment system for pure lambda 
calculus from these logics, however, is problematic, as
discussed in~\cite{Baillot04lics}.
In this paper we show that shifting from usual call-by-name to
call-by-value lambda calculus allows regaining strong 
connections with the underlying logic. This will be done in the
context of Elementary Affine Logic (\EAL), designing a type system 
in natural deduction style assigning \EAL\ formulae to
lambda terms. 
\end{abstract}

\maketitle

%%%%%%%%%%%%%%%%%%%%%%
\section{Introduction}\label{sect:intro}
%%%%%%%%%%%%%%%%%%%%%%

The so-called light logics~\cite{Girard98ic,Asperti98lics,Asperti02tocl} 
have been introduced as logical counterparts of complexity classes, 
namely polynomial and elementary time functions. 
After their introduction, they have been shown to be relevant for optimal 
reduction~\cite{Coppola01tlca,Coppola03tlca}, programming language 
design~\cite{Asperti02tocl,Terui01lics} and set theory~\cite{Terui02phd}. 
However, proof languages for these logics, designed through the 
Curry-Howard correspondence, are syntactically quite complex
and can hardly be proposed as programming languages. An 
interesting research challenge is the
design of type systems assigning light logics formulae to 
pure lambda-terms,
forcing the class of typable terms to enjoy the same remarkable 
properties which can be proved for the logical systems. 
The mismatch between $\beta$-reduction in the lambda-calculus 
and cut-elimination in logical systems, however,
makes it difficult to both getting the subject reduction property and 
inheriting the complexity properties from the logic, as discussed 
in~\cite{Baillot04lics}. 
Indeed, $\beta$-reduction is more permissive than the 
restrictive copying
discipline governing calculi directly derived from
light logics. Consider, for example, the following
expression in $\Lambda_{\mathsf{LA}}$ (see~\cite{Terui01lics}):
$$
\mathit{let}\;\;M\;\;\mathit{be}\;\;{!x}\;\;\mathit{in}\;\;N
$$
This rewrites to $N\{x/P\}$ if $M$ is $!P$,
but is not a redex if $M$ is, say, an
application. 
It is not possible to map this mechanism into pure
lambda calculus. The solution proposed by 
Baillot and Terui~\cite{Baillot04lics} in the context of Light
Affine Logic (\LAL, see~\cite{Asperti98lics,Asperti02tocl}) 
consists in defining a type-system which is 
strictly more restrictive than the one 
induced by the logic. In this way, they both achieve subject
reduction and a strong notion of polynomial time soundness.\par

Now, notice that mapping the above let expression to 
the application 
$$
(\lambda x.N)M
$$ 
is not meaningless if we shift from the usual call-by-name lambda calculus 
to the call-by-value lambda calculus, where $(\lambda x.N)M$
is not necessarily a redex. 
In this paper, we make the best of this idea, introducing a type
assignment system, that we call \ETAS, assigning formulae of 
Elementary Affine Logic (\EAL) to lambda-terms.
\ETAS\ enjoys the following 
remarkable properties:
\begin{enumerate}[$\bullet$]
  \item
  The language of types coincides with the language of \EAL\ formulae.
  \item
  Every proof of \EAL\ can be mapped into
  a type derivation in \ETAS.
  \item 
  (Call-by-value) subject reduction holds.
  \item
  Elementary bounds can be given on the length of any
  reduction sequence involving a typable term. A similar
  bound holds on the size of terms involved in the reduction.
  \item
  Type inference is decidable and the principal typings can be inferred in polynomial time.
\end{enumerate}
The basic idea underlying \ETAS\ consists in partitioning
premises into three classes, depending on whether they are used 
once, or more than once, or they are in an intermediate status.
We believe this approach can work for other light logics too, and 
some hints will be given.

The proposed system is the first one satisfying the above properties
for light logics. A notion of typability for lambda calculus has 
been defined in~\cite{Coppola01tlca,Coppola03tlca,Baillot05tlca} 
for \EAL, and in~\cite{Baillot02ifip} for \LAL. 
Type inference has been proved to be decidable. In both cases, 
however, the notion of typability is not 
preserved by $\beta$-reduction. 

Noticeably, the proposed approach can be extended to Light
Affine Logic and Soft Affine Logic (\SAL, see~\cite{Baillot04fossacs,Lafont02SLL}).

A preliminary version of the present paper is \cite{Cop-DLag-Ron:EALCBV-05}: here some results have been
improved. In particular a new type inference algorithm is presented, and its complexity is analyzed: it turns out that
our type inference algorithm for \EAL\ has a complexity of the same order than the type inference for simple types.
Moreover some discussions about possible extensions of this method
have been added.

The paper is organized as follows: in Section~\ref{sect:existingworks} a comparison with existing work is made,
in Section~\ref{sect:preliminaries} some preliminary 
notions about \EAL\ and lambda calculus are recalled,
in Section~\ref{sect:ETAS} the \ETAS\ system is introduced, and in Section~\ref{sect:BOUNDS} and~\ref{sect:typeinf} its main properties, namely complexity bounds and a type inference algorithm, are explained. 
Section~\ref{sect:completeness} presents two possible extensions, allowing to
reach completeness for elementary functions,  
and in Section~\ref{sect:otherlightlogics} some hints on how to apply our idea to other light logics are given.
Section~\ref{sect:conclusions} contains a short summary of the obtained results.
 
%%% Local Variables: 
%%% mode: latex
%%% TeX-master: "llcbv"
%%% End: 

%%%%%%%%%%%%%%%%%%%%%%
\section{Comparison with Existing Work}\label{sect:existingworks}
%%%%%%%%%%%%%%%%%%%%%%
This work is not the first contribution on type systems derived from
light logics. We should mention works on (principal) type inference for
Elementary Affine Logic and Light Affine Logic by Baillot, Coppola,
Martini and Ronchi Della Rocca~\cite{Coppola01tlca,Baillot02ifip,Coppola03tlca}.
There, the goal was basically proving decidability of type inference. The proposed type
systems were the ones directly induced from logical systems. Typable lambda
terms can be efficiently reduced using Lamping's abstract algorithm, although
basic properties like subject reduction and complexity bounds were
not necessarily verified.

Baillot and Terui~\cite{Baillot04lics} proposed
a type system inspired by light logics and 
enjoying subject reduction and polynomial time normalization, called
Dual Light Affine Logic (\DLAL). The underlying 
term system is ordinary lambda-calculus with usual, call-by-name reduction. 
They've recently proved~\cite{Atassi06} that system $\mathsf{F}$ terms can be 
decorated with light types in polynomial time, following similar work
for Elementary Affine Logic~\cite{Baillot05tlca}.

Our approach should be understood as complementary to the one
proposed by Baillot
and Terui~\cite{Baillot04lics}: we exploit call-by-value evaluation
and this allows us to stay closer to logical systems. On the other
hand, the way our type system is formulated prevents us from getting
the full power of second-order quantification. Nevertheless, second-order
quantification is not as crucial with call-by-value as with usual call-by-name,
where data can be encoded in Church-style, following 
Berarducci and B\"ohm~\cite{Bohm85tcs}.
%%% Local Variables: 
%%% mode: latex
%%% TeX-master: "llcbv"
%%% End: 

%%%%%%%%%%%%%%%%%%%%%%%
\section{Preliminaries}\label{sect:preliminaries}
%%%%%%%%%%%%%%%%%%%%%%%
In this section we recall the proof calculus for Elementary Affine Logic, $\LambdaEA$. 
Then relations with the lambda calculus will be discussed.  

\begin{defi}[Terms, Types, Contexts]\hfil
\begin{enumerate}[i)] 
 \item The set $\Lambda$ of terms of the lambda calculus is
  defined by the grammar $M::= x \mid MM \mid \lambda x.M$, where 
  $x\in\mathit{Var}$, a countable set of variables.  
\item The grammar generating the set $\LambdaEA$ of terms of the 
  Elementary Lambda Calculus (\emph{\ea-terms} for short)
  is obtained from the previous one by adding rules 
  $$
  M ::=\;\promote{M}{M}{x}{M}{x} \mid \contr{M}{M}{x}{y}
  $$
  and by constraining all variables to occur at most once. These
  two constructs interpret promotion and contraction, respectively.
\item \emph{\ea-types} are formulae of 
  (Propositional) Elementary Affine Logic (hereby \EAL),
  and are generated by the grammar $A ::= a\mid A
  \lin A \mid\;!A$ where $a$ belongs to a countable set of basic
  type constants.  \ea-types will be ranged over by $A,B,C$.
\item \emph{\ea-contexts} are finite subsets of \ea-type
  assignments to variables, where all variables are different.  
  Contexts are ranged over by $\Phi$,
  $\Psi$. If $\Phi = \{x_{1}:A_{1},\ldots,x_{n}:A_{n}\}$, then $dom
  (\Phi) = \{x_{1},\ldots,x_{n}\} $.  Two contexts are \emph{disjoint}
  if their domains have empty intersection.
\item The type assignment system in natural-deduction style for \ea-terms
  ($\vdashNEAL$ for short) assigns \ea-types
  to \ea-terms. The system is given in Table \ref{tab:neal-type-system}. With a slight
  abuse of notation, we will denote by \NEAL\ the set of typable
  terms in $\LambdaEA$.
\end{enumerate}
\end{defi}
  
\noindent Both $\LambdaEA$ and $\Lambda$ 
are ranged over by $M,N,P,Q$. The context should help avoiding
ambiguities.  Symbol $\equiv$ denotes syntactic identity on both types
and terms. The identity on terms is taken modulo names of bound
variables and modulo permutation in the list
${}^{M}/x,\cdots,{}^{M}/x$ inside $!\left(M\right)[{}^{M}/x,\ldots,$
${}^{M}/x]$.
%and in contracted variables $x,y$ inside $\contr{M}{M}{x}{y}$.  
  
\begin{table*}
\begin{center}
\fbox{
\begin{minipage}{.97\textwidth}
\vspace{3mm}
$$
\begin{array}{rcl}
\infer[A]{\Phi,x:A\vdashNEAL x:A}{} 
&\hspace{3mm} &
\infer[C]{\Phi,\Psi\vdashNEAL
  \contr{N}{M}{x}{y}:B}{\Phi\vdashNEAL M:!A &
  \Psi,x:!A,y:!A\vdashNEAL N:B}
\end{array}
$$
$$ 
\begin{array}{rcl}
\infer[I_\lin]{\Phi\vdashNEAL \lam
  x.M:A\lin  B}{\Phi,x:A\vdashNEAL M:B}
&\hspace{1mm} &
\infer[E_\lin]{\Phi,\Psi\vdashNEAL M\ N:B}{\Phi\vdashNEAL M:A\lin B &
  \Psi\vdashNEAL N:A}
\end{array}
$$
$$
\infer[!]{\Phi,\Psi_1,\ldots,\Psi_n\vdashNEAL
  \promote{N}{M_1}{x_1}{M_n}{x_n}:!B}{%
  \Psi_1\vdashNEAL M_1:!A_1\quad\cdots\quad
  \Psi_n\vdashNEAL M_n:!A_n& 
  x_1:A_1,\ldots,x_n:A_n\vdashNEAL N:B}
$$
\vspace{0mm}
\end{minipage}}
\end{center}   
\caption{Type assignment system for \ea-terms.
Contexts with different names are intended to be disjoint.}\label{tab:neal-type-system}
\end{table*}

\begin{table*}
\begin{center}
\fbox{
%\scalebox{.9}{%
  \begin{minipage}{0.97\textwidth}
 \begin{displaymath}
 \begin{array}{l}
 (\lam x.M\ N) \qquad \to_\beta \qquad M\{N/x\}\\\\
 \contr{N}{\promote{M}{M_1}{x_1}{M_n}{x_n}}{x}{y} \qquad
 \to_{\mathsf{dup}} \\ 
 \qquad \contr{\ldots\contr{N\{ ^{\promote{M}{x_1'}{x_1}{x_n'}{x_n}}/x\}
     \{ ^{\promote{M'}{y_1'}{y_1}{y_n'}{y_n}}/y\}}{M_1}{x_1'}{y_1'}
   \cdots}{M_n}{x_n'}{y_n'}\\\\
 !(M)[ ^{M_1}/x_1,\cdots,
 ^{\promote{N}{P_1}{y_1}{P_m}{y_m}}/x_i,\cdots, ^{M_n}/x_n] \qquad
 \to_{!-!}\\ 
 \qquad  !(M\{N/x_i\})[ ^{M_1}/x_1,\cdots, ^{P_1}/y_1, \cdots, ^{P_m}/y_m,
 \cdots ^{M_n}/x_n]\\\\
 (\contr{M}{M_1}{x_1}{x_2}\ N) \qquad \to_{@-\mathsf{c}}
 \qquad \contr{(M\{x_1'/x_1,x_2'/x_2\}\ N)}{M_1}{x_1'}{x_2'}\\\\
 (M\ \contr{N}{N_1}{x_1}{x_2}) \qquad \to_{@-\mathsf{c}}
 \qquad \contr{(M\ N\{x_1'/x_1,x_2'/x_2\})}{N_1}{x_1'}{x_2'}\\\\
 !(M)[ ^{M_1}/x_1,\cdots,
 ^{\contr{M_i}{N}{y}{z}}/x_i,\cdots, ^{M_n}/x_n] \qquad
 \to_{!-c}\\
 \qquad \contr{!(M)[ ^{M_1}/x_1,\cdots, ^{M_i\{y'/y,z'/z\}}/x_i,\cdots,
   ^{M_n}/x_n]}{N}{y'}{z'}
 \end{array}
 \end{displaymath}
 \begin{displaymath}
 \begin{array}{l}
 \contr{M}{\contr{N}{P}{y_1}{y_2}}{x_1}{x_2} \qquad
 \to_{\mathsf{c}-\mathsf{c}} \qquad
 \contr{\contr{M}{N\{y_1'/y_1,y_2'/y_2\}}{x_1}{x_2}}{P}{y_1'}{y_2'}\\\\
 \lam x.\contr{M}{N}{y}{z} \qquad \to_{\lam-\mathsf{c}} \qquad
 \contr{\lam x.M}{N}{y}{z} \mbox{\textrm{ where }}
 x\notin\mathtt{FV}(N)
 \end{array}
 \end{displaymath}
 where $M'$ in the $\to_{\mathsf{dup}}$-rule is obtained from $M$
 replacing all its free variables with fresh ones ($x_i$ is replaced
 with $y_i$); $x_1'$ and $x_2'$ in the $\to_{@-\mathsf{c}}$-rule, $y'$
 and $z'$ in the $\to_{!-c}$-rule and $y_1',y_2'$ in the
 $\to_{\mathsf{c}-\mathsf{c}}$-rule are fresh variables.
  \end{minipage}}
%}
\end{center}
\caption{Normalization rules in $\Lambda^{\ea}$.}\label{fig:ea-reduction} 
\end{table*}

On $\Lambda$, both the call-by-name and the call-by-value
$\beta$-reduction will be used, according to the following definition.
\vfill\eject

\begin{defi}[Reduction]\hfill
\begin{enumerate}[i)]
\item We refer to the contextual
closure of the rule $(\lambda x. M) N\!\bn\! M\{ N/x\}  $, where $M\{ N/x\}$ denotes the
capture free substitution of $N$ to the free occurrences of $x$ in
$M$, as the \emph{call-by-name $\beta$-reduction};
\item \emph{Values} are generated by the grammar
$V::=x\mid\lambda x.M$ where $x$ ranges over $\mathit{Var}$ and $M$ ranges
over $\Lambda$. $\mathcal{V}$ is the set of all values. Values are denoted by $V,U,W$.
 The {\em call-by-value} $\beta$-reduction is the contextual closure of the
rule $(\lambda x. M) V \bv M\{ V/x\}$ where $V$ ranges over values.
\item Let $t \in \{n,v \}$; symbols $\rightarrow^{+}_{t}$ and $\rightarrow^{*}_{t}$
denote the transitive closure and the symmetric and transitive closure of
$\rightarrow_{t}$, respectively.            
\end{enumerate}  
\end{defi}

\noindent A term in $\LambdaEA$ can be transformed naturally to a term in $\Lambda$ by performing
the substitutions which are explicit in it, and forgetting the modality $!$. Formally,
the translation function $(\cdot)^*:\LambdaEA\to\Lambda$ is
defined by induction on the structure of \ea-terms as follows:
\begin{eqnarray*}
  (x)^* &=& x \\
  (\lam x.M)^* &=& \lam x.(M)^* \\
  (M N)^* &=& (M)^* (N)^*\\
  (\contr{M}{N}{x_1}{x_2})^* &=& (M)^*\{(N)^*/x_1,(N)^*/x_2\} \\
  (\promote{N}{M_1}{x_1}{M_n}{x_n})^* &=&
    (N)^*\{(M_1)^*/x_1,\ldots,(M_n)^*/x_n\}
\end{eqnarray*} 
where $M\{{}^{M_1}/x_1,\cdots,{}^{M_n}/x_n\}$ denotes the
simultaneous substitution of all free occurrences of $x_i$ by $M_i$ ($1\leq i \leq n$).

The map $(\cdot)^*$ easily induces a type-assignment system $\NEAL^*$ for
pure lambda-calculus: take \NEAL\ and replace every occurrence of
a term $M$ by $M^*$ in every rule. Normalization in \NEAL (see Table~\ref{fig:ea-reduction}), however,
is different from normalization in lambda-calculus --- $\NEAL^*$ does not even satisfy subject-reduction. Moreover,
lambda calculus does not provide any mechanism for sharing:
the argument is duplicated as soon as $\beta$-reduction fires. This,
in turn, prevents from analyzing normalization in the lambda calculus
using the same techniques used in logical systems. This phenomenon
has catastrophic consequences in the context of Light Affine Logic,
where polynomial time bounds cannot be transferred from the logic
to pure lambda-calculus~\cite{Baillot04lics}.

Consider now a different translation $(\cdot)^\#:\LambdaEA\to\Lambda$:
\begin{eqnarray*}
  (x)^\# &=& x\\
  (\lam x.M)^\# &=& \lam x.(M)^\#\\ 
  (M N)^\# &=& (M)^\#\ (N)^\#\\
  (\contr{N}{M}{x}{y})^\# &=& 
    \left\{
      \begin{array}{ll}
        (N)^\#\{M/x,M/y\} & \mbox{\emph{if $M$ is a variable}} \\
        (\lambda z.(N)^\#\{z/x,z/y\})(M)^\# & \mbox{\emph{otherwise}}
      \end{array}
\right. \\
  (\promote{N}{M_1}{x_1}{M_n}{x_n})^\# &=&
    \left\{
      \begin{array}{l}
        (N)^\#\hspace{.5cm}\mbox{\emph{if $n=0$}} \\
        (\promote{N}{M_2}{x_2}{M_n}{x_n})^\#\{M_1/x_1\}\\
        \hspace{.5cm}\mbox{\emph{if $n\geq 1$ and $M_1$ is a variable}} \\
        (\lambda x_1.(\promote{N}{M_2}{x_2}{M_n}{x_n})^\#)(M_1)^\#\\
        \hspace{.5cm}\mbox{\emph{if $n\geq 1$ and $M_1$ is not a variable}} 
      \end{array}
\right.
\end{eqnarray*}
Please observe that while $(\cdot)^\#$ maps $\contr{N}{M}{x}{y}$
and $\promote{N}{M_1}{x_1}{M_n}{x_n}$ to applications (except when the
arguments are variables), $(\cdot)^*$
maps them to terms obtained by substitution. Indeed, if lambda calculus 
is endowed with ordinary $\beta$-reduction, the two translations are 
almost equivalent:
\begin{lem}
For every \ea-term $M$, $(M)^\# \bn^{*}  (M)^*$.
\end{lem}
\begin{proof}
By induction on $M$.
\end{proof}
However, it is certainly not true that $(M)^\# \bv^{*}  (M)^*$.
The map $(\cdot)^\#$, differently from $(\cdot)^*$, does not cause
an exponential blowup on the length of terms.
The \emph{length} $L(M)$ of a term $M$ is defined
inductively as follows: 
\begin{eqnarray*}
  L(x) &=& 1 \\
  L(\lam x.M) &=& 1 + L(M)\\
  L(M\ N) &=& 1 + L(M) + L(N)
\end{eqnarray*}
The same definition can be extended to $\ea$-terms by way
of the following equations:
\begin{eqnarray*}
  L(\promotewa{M}) &=& L(M) +1\\
  L(\promote{M}{M_1}{x_1}{M_n}{x_n}) &=&
    L(\promote{M}{M_1}{x_1}{M_{n-1}}{x_{n-1}})+
    L(M_n) + 1\\
  L(\contr{M}{N}{x}{y}) &=& L(M) + L(N) +1
\end{eqnarray*}  

\begin{prop}\label{prop:corrispsize}
For every $N\in\LambdaEA$, $L(N^\#)\leq 2L(N)$.
\end{prop}
\begin{proof}
By induction on $N$. The cases for variables, abstractions
and applications are trivial. Let us now consider the
other two inductive cases. Suppose 
$N=\contr{P}{Q}{x}{y}$. If $Q$ is a variable, then 
$L(N^\#)=L(P^\#)\leq 2L(P)\leq 2L(N)$. If $Q$ is not
a variable, then $L(N^\#)=L(P^\#)+L(Q^\#)+2\leq
2L(P)+2L(Q)+2=2(L(P)+L(Q)+1)=2L(N)$.
If, on the other hand, $N=\promote{M}{M_1}{x_1}{M_n}{x_n}$,
then we can proceed by induction on $n$. If $n=0$, then
the inequality is trivially verified. If, on the other hand,
$n>0$, then we must distinguish two different cases:
if $M_n$ is a variable, then the
inequality is trivially satisfied; if $M_n$ is not a variable,
then $N^\#$ is 
$(\lambda x_n.(\promote{M}{M_1}{x_1}{M_{n-1}}{x_{n-1}})^\#))M_n^\#$
and, by the induction hypothesis on $n$ and $M_n$, we get
\begin{eqnarray*}
L(N^\#)&=&2+L((\promote{M}{M_1}{x_1}{M_{n-1}}{x_{n-1}})^\#)+
L(M_n^\#)\\&\leq& 2+2L(\promote{M}{M_1}{x_1}{M_{n-1}}{x_{n-1}})+
2L(M_n^\#)\\ &=&2L(N)
\end{eqnarray*}
This concludes the proof.
\end{proof}
%%% Local Variables: 
%%% mode: latex
%%% TeX-master: "llcbv"
%%% End: 

%%%%%%%%%%%%%%%%%%%%%%%
\section{The Elementary Type Assignment System}\label{sect:ETAS}
%%%%%%%%%%%%%%%%%%%%%%%
In this section we will define a type assignment system typing
lambda-terms with \EAL\ formulae.  We want the system to be
\emph{almost} syntax directed, the difficulty being the
handling of $C$ and $!$ rules. 
This is solved by splitting the context into three parts,
the \emph{linear} context, the \emph{modal} context, and the
\emph{parking} context. In particular the parking context is used to
keep track of premises which must become modal in the future.

\begin{defi}\hfill
\begin{enumerate}[i)]
\item An \EAL\ formula $A$ is {\em modal} if 
  $A\equiv !B$ for some $B$, it is {\em linear} otherwise.
\item A context is a set of pairs $x:A$ where $x$ is a variable and
  $A$ is an \ea-type, where all variables are disjoint. A context is {\em linear} if it assigns linear
  \ea-types to variables, while it is {\em modal} if it assigns modal
  \ea-types to variables. If $\Phi$ is a context, $\Phi^{L} $ and
  $\Phi^{I} $ denote the linear and modal sub-contexts of $\Phi$,
  respectively.
\item  The Elementary Type Assignment System (\ETAS) proves
  statements like $\Gamma \mid \Delta \mid \Theta \vdash M:
  A$ where $\Gamma$ and $\Theta$ are linear contexts and $\Delta$ is
  a modal context. The contexts have disjoint variables. The rules of the system are shown in
  Table~\ref{tab:lambda-type-system}.  In what follows,
  $\Gamma $, $\Delta$ and $\Theta$ will range over linear,
  modal and parking contexts respectively.
\item  A \emph{typing judgement} for $M$ is a statement of the kind $\Gamma \mid \Delta \mid \emptyset
  \vdash M: A$. A term $M \in \Lambda$ is \emph{\ea-typable} if there is a typing for it.
  Type derivations built according to rules in Table~\ref{tab:lambda-type-system} will
  be denoted with greek letters like $\pi$, $\rho$ and $\sigma$. If $\pi$ is a type derivation
  with conclusion $\Gamma \mid \Delta \mid \Theta\vdash M: A$, we write
  $\pi:\Gamma \mid \Delta \mid \Theta\vdash M: A$.
%\item 
%  Given a type derivation $\pi$, its \emph{skeleton} is an ordered tree obtained from $\pi$
%  by erasing all typing information except rule name. Two type derivations are said to be 
%  \emph{structurally equivalent} if they have identical skeletons.
\end{enumerate}  
\end{defi}  

\begin{table*} 
\begin{center}
\fbox{
\begin{minipage}{.95\textwidth}
\vspace{3mm}
$$
\begin{array}{rcl}
\infer[A^L]{\Gamma, x:A\mid\Delta\mid\Theta\vdash x:A}{}
&\hspace{8mm}& 
\infer[A^P]{\Gamma\mid\Delta\mid x:A,\Theta\vdash x:A}{}
\end{array}
$$
$$
\begin{array}{rcl}
\infer[I_\lin^L]{\Gamma\mid\Delta\mid\Theta\vdash\lam x.M:A\lin B}{
\Gamma,x:A\mid\Delta\mid\Theta\vdash M: B}
&\hspace{6mm}&
\infer[I_\lin^I]{\Gamma\mid\Delta\mid\Theta\vdash\lam x.M:A\lin B}{
\Gamma\mid\Delta,x:A\mid\Theta\vdash M: B}
\end{array}
$$
$$
\infer[E_\lin]{\Gamma_1,\Gamma_2\mid\Delta\mid\Theta\vdash M\ N: B}{
\Gamma_1\mid\Delta\mid\Theta\vdash M:A\lin B 
&
\Gamma_2\mid\Delta\mid\Theta\vdash N:A}
$$
$$
\infer[!]{\Gamma_2\mid !\Gamma_1,!\Delta_1,!\Theta_1,
\Delta_2\mid\Theta_2\vdash M:!A}{\Gamma_1\mid\Delta_1\mid\Theta_1\vdash M:A}
$$
\vspace{0mm}
\end{minipage}}
\end{center}
\caption{The Elementary Type Assignment System (\ETAS). Contexts with different names
are intended to be disjoint.  }\label{tab:lambda-type-system}
\end{table*}

\noindent Rules $A^L$ and $A^P$ (see Table~\ref{tab:lambda-type-system})
are two variations on the classical axiom
rule. Notice that a third axiom rule
$$
\infer[A^I]{\Gamma\mid x:!A,\Delta\mid\Theta\vdash x:!A}{}
$$
is derivable. Abstractions cannot be performed on variables
in the parking context. The rule $E_\linear$ is the standard rule 
for application. Rule $!$ is 
derived from the one traditionally found in sequent calculi and
is weaker than the rule induced by \NEAL\ via $(\cdot)^*$.
Nevertheless,
it is sufficient for our purposes and (almost) syntax-directed.
The definition of an \ea-typable term takes into account the 
auxiliary role of the parking context. 

\begin{exa}\label{exa:terms}
Let us illustrate the r\^oles of the various \ETAS\ rules by way of an
example. Consider the Church's numeral $\underline{2}\equiv\lambda x.\lambda y.x(xy)$,
let $B$ be $!(A\multimap A)$ and $C$ be $B\multimap B$. A type derivation for
$\underline{2}$ is the following:
{\footnotesize
$$
\infer[I_\lin^I]
  {\emptyset\mid\emptyset\mid\emptyset\vdash \lambda x.\lambda y.x(xy):!C\multimap !C}
  {
    \infer[!]
    {\emptyset\mid x:!C\mid\emptyset\vdash \lambda y.x(xy):!(B\multimap B)}
    {
      \infer[I_\lin^I]
      {\emptyset\mid\emptyset\mid x:C\vdash \lambda y.x(xy):B\multimap B}
      {
        \infer[E_\lin]
        {\emptyset\mid y:B\mid x:C\vdash x(xy):B}
        {
          \infer[\!E_\lin]
          {\emptyset\mid y:B\mid x:C\vdash xy:B}
          {
            \infer[A^P]
            {\emptyset\mid y:B\mid x:C\vdash x:C}
            {}
            &
            \infer[!]
            {\emptyset\mid y:B\mid x:C\vdash y:B}
            {
              \infer[A^L\!]
              {y:A\multimap\!A\mid\emptyset\mid\emptyset\vdash y:A\multimap\!A}
              {}
            }
          }
          &
          \infer[\!A^P]
          {\emptyset\mid y:B\mid x:C\vdash x:C}
          {}
        }
      }
    }
  }
$$
}%
Call this type derivation $\pi(2,B)$.
But Church numerals can be typed slightly differently. Consider the
term $\underline{3}\equiv\lambda x.\lambda y.x(x(xy))$ and the
following type derivation (where $D$ stands for $A\lin A$):
{\footnotesize
$$
\infer[I_\lin^I]
  {\emptyset\mid\emptyset\mid\emptyset\vdash \lambda x.\lambda y.x(x(xy)):C}
  {
    \infer[!]
    {\emptyset\mid x:B\mid \emptyset\vdash \lambda y.x(x(xy)):B}
    {
      \infer[I_\lin^I]
      {\emptyset\mid\emptyset\mid x:D\vdash \lambda y.x(x(xy)):D}
      {
        \infer[E_\lin]
        {y:A\mid\emptyset\mid x:D\vdash x(x(xy)):A}
        {
          \infer[E_\lin]
          {y:A\mid\emptyset\mid x:D\vdash x(xy):A}
          {
            \infer[E_\lin]
            {y:A\mid\emptyset\mid x:D\vdash xy:A}
            {
              \infer[A^P]
              {\emptyset\mid\emptyset\mid x:D\vdash x:D}
              {}
              &
              \infer[A^L]
              {y:A\mid\emptyset\mid x:D\vdash y:A}
              {}
            }
            &
            \infer[A^P]
            {\emptyset\mid\emptyset\mid x:D\vdash x:D}
            {}
          }
          &
          \infer[A^P]
          {\emptyset\mid\emptyset\mid x:D\vdash x:D}
          {}
        }
      }
    }
  }
$$}
This is $\pi(3,A)$
This way we can give the application $\underline{2}\;\underline{3}$ the
type $!C$.
\end{exa}

This system does not satisfy call-by-name subject-reduction. 
Consider, for example, the lambda term
$M \equiv (\lambda x.yxx)(wz)$. 
A typing for it is the following:
$$y:!A\multimap !A\multimap A,w:A \multimap!A,z:A
\;|\;\emptyset\;|\;\emptyset \vdash M: A$$
$M \bn N$, where $N\equiv y(wz)(wz)$ and
$y:!A\multimap !A\multimap A,w:A \multimap!A,z:A
\;|\;\emptyset\;|\;\emptyset \not\vdash N: A $, because 
rule $E_\lin$ requires the two linear contexts to be disjoint. 
%In fact the rule enforces
%a variable occurring more than once to have a linear type.  
Note that
both $\emptyset \;|\;\emptyset\;|\;y:!A\multimap !A\multimap A,w:A
\multimap!A,z:A \vdash M: A$ and $\emptyset
\;|\;\emptyset\;|\;y:!A\multimap !A\multimap A,w:A \multimap!A,z:A
\vdash N:A$, but 
%this does not imply $N$ to be
these are not 
\ea-typings. 
%Moreover, $\lambda w.M\bn\lambda w.N$, but
%while $M$ can be given type $(A\lin !A)\lin A$, $N$ cannot.

The subject reduction problem, however, disappears when
switching from call-by-name to call-by-value reduction.

\begin{lem}[Weakening Lemma]\label{lemma:weakening}
If $\pi:\Gamma_1\mid\Delta_1\mid\Theta_1\vdash M:A$, then there is
$\sigma:\Gamma_1,\Gamma_2\mid\Delta_1,\Delta_2\mid\Theta_1,\Theta_2\vdash M:A$,
for every $\Gamma_{2}, \Delta_{2}, \Theta_{2}$ disjoint from each other and from 
$\Gamma_{1}, \Delta_{1}, \Theta_{1}$.
Moreover, the number of rule instances in $\sigma$ is identical to the number of rule
instances in $\pi$.
\end{lem}
\begin{lem}[Shifting Lemma]\label{lemma:lintopark}
If $\pi:\Gamma,x:A\mid\Delta\mid\Theta\vdash M:B$, then there is
$\sigma:\Gamma\mid\Delta\mid x:A,\Theta\vdash M:B$.
Moreover, the number of rule instances in $\sigma$ is identical to the number of
rule instances in $\pi$.
\end{lem}
\begin{lem}[Substitution Lemma]\label{lemma:substitution}
Suppose $\Gamma_{1} $ and $\Gamma_{2} $ are disjoint contexts. Then:
\begin{enumerate}[\em i)]
  \item
  If $\pi:\Gamma_1,x:A\mid\Delta\mid\Theta\vdash M:B$ and
  $\sigma:\Gamma_2\mid\Delta\mid\Theta\vdash N:A$, then there is
  $\rho:\Gamma_1,\Gamma_2\mid\Delta\mid\Theta\vdash M\{N/x\}:B$.
  \item
  If $\pi:\Gamma\mid\Delta\mid x:A,\Theta\vdash M:B$ and 
  $\sigma:\emptyset\mid\Delta\mid\Theta\vdash N:A$, then there is
  $\rho:\Gamma\mid\Delta\mid \Theta\vdash M\{N/x\}:B$.
  \item
  If $\pi:\Gamma_1\mid\Delta,x:A\mid\Theta\vdash M:B$, 
  $\sigma:\Gamma_2\mid\Delta\mid\Theta\vdash N:A$ and $N\in \mathcal{V}$, then
  there is $\rho:\Gamma_1,\Gamma_2\mid\Delta\mid\Theta\vdash M\{N/x\}:B$.
\end{enumerate}
\end{lem}
\begin{proof}
  The first point can be easily proved by induction on the derivation
  for $\Gamma_1,x:A\mid\Delta\mid\Theta\vdash M:B$ using, in
  particular, the Weakening Lemma. 

  Let us prove the second point (by the same
  induction). The case for $A^P$ can be proved by way of the previous
  lemmas. $I^L_\lin$ and $I^I_\lin$ are trivial. $E_\lin$ comes
  directly from the induction hypothesis and
  Lemma~\ref{lemma:weakening}.  $!$ is trivial since $x$ cannot appear
  free in $M$ and so $M\{N/x\}$ is just $M$.

  The third point can
  be proved by induction, too, but it is a bit more difficult. First
  of all, observe that $A$ must be in the form
  $\underbrace{!...!}_{n}C$, with $n\geq 1$. Let us focus on rules
  $E_\lin$ and $!$ (the other ones can be handled easily). Since $N\in\mathcal{V}$,
  the derivation for $\Gamma_2\mid\Delta\mid\Theta\vdash N:A$ must end
  with $A^L$, $A^P$, $I^L_\lin$ or $I^I_\lin$ (depending on the shape
  of $N$), followed by exactly $n$ instances of the $!$ rule, being it
  the only non-syntax-directed rule. If the last rule used in $\pi$
  is $E_\lin$, then $\pi$ has the following shape:
  $$
  \infer{\Gamma_1\mid x:A,\Delta\mid\Theta\vdash M:B}
  {
    \phi:\Gamma_3\mid x:A,\Delta\mid\Theta\vdash L:D\lin B &
    \psi:\Gamma_4\mid x:A,\Delta\mid\Theta\vdash P:D
  }
  $$
  where $\Gamma_1\equiv\Gamma_3,\Gamma_4$ and $M\equiv LP$. 
  $\sigma$ can be written as follows:
  $$
  \infer{\Gamma_2\mid\Delta\mid\Theta\vdash N:A}
  {
    \xi:\Gamma_5\mid\Delta_1\mid\Theta_1\vdash N:C &
  }
  $$
  where $\Delta\equiv !\Gamma_5,!\Delta_1,!\Theta_1,
  \Delta_2$ and $A\equiv !C$. From $\xi$ we can obtain
  a derivation $\chi:\emptyset\mid\Delta\mid\Theta\vdash N:A$
  and applying (two times) the induction hypothesis, we get
  $\mu:\Gamma_3\mid\Delta\mid\Theta\vdash L\{N/x\}:D\lin B$
  and $\nu:\Gamma_4\mid\Delta\mid\Theta\vdash P\{N/x\}:D$
  from which we get the desired $\rho$ by applying rule
  $E_\lin$ and Lemma~\ref{lemma:weakening}.
  If the last rule used in $\pi$ is $!$, then $\pi$ has the following shape:
  $$
  \infer{\Gamma_1\mid x:A,\Delta\mid\Theta\vdash M:B}
  {
    \phi:\Gamma_3\mid\Delta_1\mid\Theta_1\vdash M:C &
  }
  $$
  where $x:A,\Delta\equiv !\Gamma_3,!\Delta_1,!\Theta_1,\Delta_3$ and $B\equiv !C$.
  $\sigma$ can be written as follows:
  $$
  \infer{\Gamma_2\mid\Delta\mid\Theta\vdash N:A}
  {
    \psi:\Gamma_4\mid\Delta_2\mid\Theta_2\vdash N:D &
  }
  $$
  where $\Delta\equiv !\Gamma_4,!\Delta_2,!\Theta_2,\Delta_4$ and $A\equiv !D$.
  We now distinguish some cases:
  \begin{enumerate}[$\bullet$]
    \item
      If $x\in\mathit{dom}(\Delta_3)$, then $x\notin\mathit{FV}(M)$
      and $\rho$ is obtained easily from $\phi$.
    \item
      If $x\in\mathit{dom}(\Delta_1)$, then let $\Delta_1\equiv x:D,\Delta_5$.
      By applying several times Lemma~\ref{lemma:weakening}
      and Lemma~\ref{lemma:lintopark} we can obtain
      type derivations 
      \begin{eqnarray*}
      \xi&:&\emptyset\mid x:D,\Delta_2\cup\Delta_5\mid\Gamma_3\cup\Gamma_4\cup\Theta_1\cup\Theta_2\vdash M:C\\
      \chi&:&\emptyset\mid\Delta_2\cup\Delta_5\mid\Gamma_3\cup\Gamma_4\cup\Theta_1\cup\Theta_2\vdash N:D
      \end{eqnarray*}
      which have the same number of rule instances as $\phi$ and $\psi$, respectively.
      By applying point ii) of this Lemma, we obtain 
      $$
      \mu:\emptyset\mid\Delta_2\cup\Delta_5\mid\Gamma_3\cup\Gamma_4\cup\Theta_1\cup\Theta_2\vdash M\{N/x\}:C
      $$
      from which $\rho$ can be easily obtained.
    \item
      If $x\in\mathit{dom}(\Theta_1)$, then let $\Theta_1\equiv x:D,\Theta_3$.
      By applying several times Lemma~\ref{lemma:weakening}
      and Lemma~\ref{lemma:lintopark} we can obtain
      type derivations 
      \begin{eqnarray*}
      \xi&:&\emptyset\mid\Delta_1\cup\Delta_2\mid x:D,\Gamma_3\cup\Gamma_4\cup\Theta_2\cup\Theta_5\vdash M:C\\
      \chi&:&\emptyset\mid\Delta_1\cup\Delta_2\mid\Gamma_3\cup\Gamma_4\cup\Theta_2\cup\Theta_5\vdash N:D
      \end{eqnarray*}
      which have the same number of rule instances as $\phi$ and $\psi$, respectively.
      By applying the inductive hypothesis, we obtain
      $$
      \mu:\emptyset\mid\Delta_1\cup\Delta_2\mid\Gamma_3\cup\Gamma_4\cup\Theta_2\cup\Theta_5\vdash M\{N/x\}:C
      $$
      from which $\rho$ can be easily obtained.
    \item
      If $x\in\mathit{dom}(\Gamma_3)$, then let $\Gamma_3\equiv x:D,\Gamma_5$.
      By applying several times Lemma~\ref{lemma:weakening}
      and Lemma~\ref{lemma:lintopark} we can obtain
      type derivations 
      \begin{eqnarray*}
      \xi&:&\emptyset\mid\Delta_1\cup\Delta_2\mid x:D,\Gamma_4\cup\Gamma_5\cup\Theta_1\cup\Theta_2\vdash M:C\\
      \chi&:&\emptyset\mid\Delta_1\cup\Delta_2\mid\Gamma_4\cup\Gamma_5\cup\Theta_1\cup\Theta_2\vdash N:D
      \end{eqnarray*}
      which have the same number of rule instances as $\phi$ and $\psi$, respectively.
      By applying the inductive hypothesis, we obtain
      $$
      \mu:\emptyset\mid\Delta_1\cup\Delta_2\mid\Gamma_4\cup\Gamma_5\cup\Theta_1\cup\Theta_2\vdash M\{N/x\}:C
      $$
      from which $\rho$ can be easily obtained.
  \end{enumerate}
This concludes the proof.
\end{proof}
%\textbf{-p. forse bisognerebbe aggiungere questo lemma: $\Gamma\mid\Delta, 
% x:!^n A\mid \Theta\vdash M:B$ then either $x\notin \FV(M)$ or there are 
% $n$ ! rules in the derivation where $x$ appears in the intuitionistic context}
\begin{thm}[Call-by-Value Subject Reduction]\label{theo:SR}
$\Gamma \mid \Delta \mid \Theta \vdash M: A$ and $M \bv N$ implies
$\Gamma \mid \Delta \mid \Theta \vdash N: A$.  
\end{thm} 
\begin{proof}
A redex is a term of the shape $(\lambda x.M')N'$, where $N'\in \mathcal{V}$. Then 
it can be the subject of a subderivation ending by an application of the rule $(E_{\lin})$ immediately 
preceded by 
an application of rule $(I_{\lin})$. So the result follows by the Substitution Lemma.
\end{proof}   

We are now going to prove 
that the set of typable $\lambda$-terms coincides with
$(\NEAL)^{\#}$. To do this we need the following
lemma.

\begin{lem}[Contraction Lemma]\label{lemma:contraction}\hfill
\begin{enumerate}[\em i)]
  \item
  If $\Gamma\mid\Delta\mid x:A,y:A,\Theta\vdash M:B$, then
  $\Gamma\mid\Delta\mid z:A,\Theta\vdash M\{z/x,z/y\}:B$
  \item
  If $\Gamma\mid x:A,y:A,\Delta\mid\Theta\vdash M:B$, then
  $\Gamma\mid z:A,\Delta\mid\Theta\vdash M\{z/x,z/y\}:B$
\end{enumerate}
\end{lem}

\begin{thm}\label{ETAS-and-NEAL}\hfill
\begin{enumerate}[\em i)] 
\item If $\Phi\vdashNEAL M:A$ then $\Phi^{L} \mid\Phi^{I}
  \mid\emptyset\vdash (M)^\#:A$.
\item  If $\Gamma\mid\Delta\mid\emptyset \vdash M:A$, there
  is $N\in\LambdaEA$ such that $(N)^\#=M$
  and $\Gamma,\Delta\vdashNEAL N:A$. 
\end{enumerate} 
\end{thm}
\proof\hfill
\begin{enumerate}[i)]
\item By induction on the structure of the derivation
for $\Phi\vdashNEAL M:A$. Let us focus on nontrivial
cases. \par
If the last used rule is $E_\lin$, the two premises are
$\Phi \vdashNEAL N:B\lin C$ and
$\Phi_2\vdashNEAL P:B$, and $M \equiv NP$. By induction hypothesis,
$\Phi^{L}_1\mid\Phi^{I}_1\mid\emptyset\vdash (N)^\#:B\lin C$, and
$\Phi^{L}_2\mid\Phi^{I}_2\mid\emptyset\vdash (P)^\#:B$ and,
by Weakening Lemma, 
$\Phi^{L}_1\mid\Phi^{I}_1,\Phi^{I}_2\mid\emptyset\vdash (N)^\#:B\lin C$,
$\Phi^{L}_2\mid\Phi^{I}_1,\Phi^{I}_2\mid\emptyset\vdash (P)^\#:B$
Rule $E_\lin$ leads to the thesis. 

If the last used rule is $C$, the two premises
are $\Phi_1\vdashNEAL N:!A$
and $\Phi_2,x:!A,y:!A\vdashNEAL P:B$. By induction
hypothesis, $\Phi^{L}_1\mid\Phi^{I}_1\mid\emptyset\vdash (N)^\#:!A$,
$\Phi^{L}_2\mid\Phi^{I}_2,x:!A,y:!A\mid\emptyset\vdash (P)^\#:B$.
By Contraction Lemma, 
$\Phi^{L}_2\mid\Phi^{I}_2,z:!A\mid\emptyset\vdash (P)^\#\{z/x,z/y\}:B$
and so
$\Phi^{L} _2\mid\Phi^{I}_2\mid\emptyset\vdash \lambda z.(P)^\#\{z/x,z/y\}:!A\lin B$
By rule $E_\lin$ and Weakening Lemma, we finally get
$\Phi^{L}_1,\Phi^{L} _2\mid\Phi^{L}_1,\Phi^{I}_2\mid\emptyset\vdash 
(\lambda z.(P)^\#\{z/x,z/y\})(N)^\#:B$.
\item 
%Let $FV(M)=\{x1,...,x_{n}\}  $, and
%let $\mathit{lin}(M)$ be a term such that every free variable in it occurs once, and
%$\mathit{lin}(M)\{x_1/y_1^1,\ldots,x_1/y_1^{m_1},\ldots,x_n/y_n^1,
%\ldots,x_n/y_n^{m_n}\}
%\equiv M$. 
The following, stronger, statement can
be proved by induction on $\pi$:
if $\pi:\Gamma\mid\Delta\mid x_1:A_1,\ldots,x_n:A_n\vdash M:A$, then
there is $N\in\LambdaEA$ such that 
$$
M=(N)^\#\{x_1^1/y_1,\ldots,x_1^{m_1}/y_1,
\ldots,x_n^1/y_n,\ldots,x_n^{m_n}/y_n\}
$$
and $\Gamma,\Delta,y_1^1:A_1,\ldots,y_1^{m_1}:A_1,$ $\ldots,y_n^1:A_n,\ldots,
y_n^{m_n}:A_n\vdashNEAL N:A$.\qed
\end{enumerate}

\noindent We have just established a deep \emph{static} correspondence
between \NEAL\ and the class of typable lambda terms. But what
about \emph{dynamics}? Unfortunately, the two systems are not
bisimilar. Nevertheless, every
call-by-value reduction-step in the lambda calculus corresponds
to at least one normalization step in $\LambdaEA$.
A normalization step in $\LambdaEA$ is denoted by $\rightarrow$; 
$\rightarrow^{+} $ denotes the transitive closure of $\rightarrow$.

An \emph{expansion} is a term in $\LambdaEA$ that can
be written either as $\promote{M}{x_1}{y_1}{x_n}{y_n}$
or as $\contr{N}{z}{x}{y}$, where $N$ is itself an
expansion.

\begin{lem}\label{lemma:expansion}
If $M$ is an expansion, then 
\begin{enumerate}[$\bullet$]
\item
$\contr{L}{M}{x}{y}\rightarrow^* P$, where $P^\#\equiv L\{M^\#/x,M^\#/y\}$;
\item
If $M\equiv P_i$, then  $\promote{L}{P_1}{x_1}{P_n}{x_n}\rightarrow^*Q$
where 
$$
Q^\#\equiv(!\left(L\right)
    \left[{}^{P_1}/x_1,\ldots,{}^{P_{i-1}}/x_{i-1},
    {}^{P_{i+1}}/x_{i+1},\ldots,{}^{P_n}/x_n\right])^\#\{M^\#/x_i\}.
$$
\end{enumerate} 
\end{lem}
\begin{prop}\label{prop:correspredseq}
  For every $M\in\LambdaEA$, if $\Gamma\vdashNEAL M:A$ and $(M)^\#\bv
  N$, then there is $L\in\LambdaEA$ such that $(L)^\#=N$ and
  $M\rightarrow^+ L$.
\end{prop}
\begin{proof}
  We can proceed by induction on the structure of $M$.  If $M$ is a
  variable, then $M^\#$ is a variable, too, and so the premise is
  false. If $M$ is an abstraction, then the thesis follows from the
  inductive hypothesis.  If $M$ is an application $P\ Q$, then we can
  assume $P$ to be an abstraction $\lambda x.R$ and $N$ to be
  $R^\#\{Q^\#/x\}$ (in all the other cases the thesis easily follows
  by induction hypothesis). It is easy to see that
  $R^\#\{Q^\#/x\}\equiv (R\{Q/x\})^\#$ and so we can take $R\{Q/x\}$
  as our $L$. If $M$ is $\contr{P}{Q}{x}{y}$ and $Q$ is not a variable
  (otherwise the thesis easily follows by inductive hypothesis), then $M^\#=(\lambda
  z.P^\#\{z/x,z/y\})Q^\#$ and we can restrict to the case where $N$ is
  $P^\#\{Q^\#/x,Q^\#/y\}$. First of all, we can observe that $Q^\#$
  must be an abstraction. This means that $Q$ is an abstraction itself
  enclosed in one or more $\promote{\cdot}{x_1}{y_1}{x_n}{y_n}$
  contexts and zero or more $\contr{\cdot}{z}{x}{y}$ otherwise $M$
  cannot be typed in \EAL. This means $Q$ is an expansion and so, by
  Lemma~\ref{lemma:expansion}, we know there must be a term $R$ such
  that $R\#\equiv P^\#\{Q^\#/x,Q^\#/y\}$, and $M\rightarrow^*R$, that
  is the thesis.  $\promote{P}{Q_1}{x_1}{Q_n}{x_n}$ can be managed in
  a similar way.
\end{proof}

\begin{rem}
Notice that Proposition~\ref{prop:correspredseq} is not a
bisimulation result. In particular, there are normalization
steps in \NEAL\ that do not correspond to anything in \ETAS.
An example is the term $(\lambda x.x)(yz)$, which
rewrites to $yz$ in \NEAL\ but is a (call-by-value) normal
form as a pure lambda-term.
\end{rem}

%%%%%%%%%%%%%%%%%%%%%%%%%%
\section{Bounds on Normalization Time}\label{sect:BOUNDS}
%%%%%%%%%%%%%%%%%%%%%%%%%%
In order to prove elementary bounds on reduction sequences, 
we need to define a refined measure on lambda terms.
We can look at a type derivation
$\pi:\Gamma\mid\Delta\mid\Theta\vdash M:A$
as a labelled tree, where every node is labelled by a rule
instance. We can give the following definition:
\begin{defi}
Let $\pi:\Gamma\mid\Delta\mid\Theta\vdash M:A$.
\begin{enumerate}[i)]
\item An occurrence of a subderivation $\rho$ of $\pi$ has \emph{level} $i$ if there are
$i$ applications of the rule $!$ in the path from the root of $\rho$ to the root
of $\pi$.
\item An occurrence of a subterm $N$ of $M$ has \emph{level} $i$ in $\pi$ if 
$i$ is the maximum level of a
subderivation of $\pi$ corresponding to the particular occurrence of $N$
under consideration (and thus having $N$ as subject).
\item The level $\partial(\pi)$ of
$\pi$ is the maximum level of subderivations of $\pi$.  
\end{enumerate}
\end{defi}
\noindent Notice that the so defined level corresponds to the notion of 
box-nesting depth in proof-nets~\cite{Asperti98lics}.
\begin{exa}
Consider the derivation $\pi(2,B)$ from Example~\ref{exa:terms}.
All the occurrences of rules $A^P$ inside
$\pi(2,B)$ have level $1$, since there is one instance of $!$ in
the path joining the root of $\pi(2,B)$ to them. The occurrence of
rule $A^L$, on the other hand, has level $2$. As a consequence
all occurrence of variables in $\underline{2}$ have either level $1$ or level
$2$ in $\pi(2,B)$. Now, 
consider the (unique) occurrence of $\lambda y.x(xy)$ in $\underline{2}$. There are
two distinct subderivations corresponding to it, one with level $0$,
the other with level $1$. As a consequence, the level of $\lambda y.x(xy)$
in $\pi(2,B)$ is $1$. Since the maximum level of subderivations of
$\pi(2,B)$ is $2$, $\partial(\pi(2,B))=2$.
\end{exa}

The length $L(M)$ of a typable lambda term $M$ does not take into account levels
as we have just defined them. The following definitions reconcile
them, allowing $L(M)$ to be ``split'' on different levels.
\begin{defi}
Let $\pi:\Gamma\mid\Delta\mid\Theta\vdash M:A$.
\begin{enumerate}[i)] 
\item $S(\pi,i) $ is defined by induction on $\pi$ as follows:  
\begin{enumerate}[$\bullet$]
\item If $\pi$ consists of an axiom, then $S(\pi,0)=1$ and
  $S(\pi,i)=0$ for every $i\geq 1$;
  \item  If the last rule in $\pi$ is $I_\linear^I$ or
  $I_\linear^L$, then 
  $S(\pi,0)=S(\rho,0)+1$
  and $S(\pi,i)=S(\rho,i)$
  for every $i\geq 1$, where $\rho$ is the immediate
  subderivation of $\pi$;
  \item  If the last rule in $\pi$ is $E_\linear$ 
  then 
  $S(\pi,0)=S(\rho,0)+S(\sigma,0)+1$
  and $S(\pi,i)=S(\rho,i)+S(\sigma,i)$
  for every $i\geq 1$, where $\rho$ and $\sigma$ are
  the immediate subderivations of $\pi$;
  \item  If the last rule in $\pi$ is $!$, then
  $S(\pi,0)=0$ and $S(\pi,i)=S(\rho,i-1)$ for
  every $i\geq 1$, where $\rho$ is the immediate
  subderivation of $\pi$. 
\end{enumerate}
\item Let $n$ be the level of $\pi$. The {\em size} of $\pi$
is  $S(\pi)= \sum_{i\leq n} S(\pi,i)$.    
\end{enumerate}  
\end{defi}
\begin{exa}
Consider again $\pi(2,B)$. By definition, $S(\pi(2,B),2)=1$,
$S(\pi(2,B),1)=5$ and $S(\pi(2,B),0)=1$.
\end{exa}
The following relates $S(\pi)$ to the size of 
the term $\pi$ types:
\begin{lem}
Let $\pi:\Gamma\mid\Delta\mid\Theta\vdash M:A$. Then,
$S(\pi)=L(M)$.    
\end{lem}

Substitution Lemma can be restated in the following way:
\begin{lem}[Weakening Lemma, revisited]\label{lemma:weakening2}
If $\pi:\Gamma_1\mid\Delta_1\mid\Theta_1\vdash M:A$, then there is
$\rho:\Gamma_1,\Gamma_2\mid\Delta_1,\Delta_2\mid\Theta_1,\Theta_2\vdash M:A$.
such that $S(\pi,i)=S(\rho,i)$ for every $i$.
\end{lem}

\begin{lem}[Shifting Lemma,revisited]\label{lemma:linktopark2}
If $\pi:\Gamma,x:A\mid\Delta\mid\Theta\vdash M:B$, then there is
$\rho:\Gamma\mid\Delta\mid x:A,\Theta\vdash M:B$
such that $S(\pi,i)=S(\rho,i)$ for every $i$.
\end{lem}

\begin{lem}[Substitution Lemma, revisited]\label{lemma:substitution2}\hfill
\begin{enumerate}[\em i)]
  \item
  If $\pi:\Gamma_1,x:A\mid\Delta\mid\Theta\vdash M:B$, 
  $\rho:\Gamma_2\mid\Delta\mid\Theta\vdash N:A$ and $N\in \mathcal{V}$, then
  there is $\sigma:\Gamma_1,\Gamma_2\mid\Delta\mid\Theta\vdash M\{N/x\}:B$
  such that $S(\sigma,i)\leq S(\rho,i)+S(\pi,i)$ for
  every $i$.
  \item
  If $\pi:\Gamma\mid\Delta\mid x:A,\Theta\vdash M:B$, 
  $\rho:\emptyset\mid\Delta\mid\Theta\vdash N:A$ and $N\in \mathcal{V}$, then
  there is $\sigma:\Gamma\mid\Delta\mid\Theta\vdash M\{N/x\}:B$
  such that $S(\sigma,i) \leq S(\pi,0)S(\rho,i)+S(\pi,i)$ for
  every $i$.
  \item
  If $\pi:\Gamma_1\mid\Delta,x:A\mid\Theta\vdash M:B$, 
  $\rho:\Gamma_2\mid\Delta\mid\Theta\vdash N:A$ and $N\in \mathcal{V}$, then
  there is $\sigma:\Gamma_1,\Gamma_2\mid\Delta\mid\Theta\vdash M\{N/x\}:B$
  such that $S(\sigma,0)\leq S(\pi,0)$ and
  $S(\sigma,i)\leq (\sum_{j\leq i}S(\pi,j))S(\rho,i)+S(\pi,i)$ for every $i\geq 1$. 
\end{enumerate}
\end{lem}
\begin{proof}
For each of the three claims, we can go by induction on the structure
of $\pi$. Here, we do not concentrate on proving the existence of
$\sigma$ (it follows from lemma~\ref{lemma:substitution}) but on
proving that $\sigma$ satisfies the given bounds.  We implicitly use
Lemma~\ref{lemma:weakening2} and Lemma~\ref{lemma:linktopark2} without
explicitly citing them.  Let us first analyze the claim i). We will
prove by induction on $\pi$ that $S(\sigma,i)\leq
S(\rho,i)\min\{1,S(\pi,0)\}+S(\pi,i)$ for every $i$ (observe that
$S(\rho,i)\min\{1,S(\pi,0)\}+S(\pi,i)\leq S(\rho,i)+S(\pi,i)$.  If
$\pi$ is just an axiom, then $\sigma$ is obtained from
%$\pi$ or 
$\rho$ by the weakening lemma and
the bound holds. If the last rule in $\pi$ is $I^L_\linear$
($I^I_\linear$), then $\sigma$ is obtained by using the inductive
hypothesis on the immediate premise $\phi$ of $\pi$ obtaining $\psi$
and then applying $I^L_\linear$ ($I^I_\linear$) to $\psi$. In both cases
\begin{eqnarray*}
S(\sigma,0)&=& S(\psi,0)+1\leq S(\rho,0)\min\{1,S(\phi,0)\}+S(\phi,0)+1\\
         &\leq& S(\rho,0)\min\{1,S(\pi,0)\}+S(\pi,0)\\
\forall i\geq 1.S(\sigma,i)&=& S(\psi,i)\leq S(\rho,i)\min\{1,S(\phi,i)\}+S(\phi,i)\\
         &\leq& S(\rho,0)\min\{1,S(\pi,i)\}+S(\pi,i)
\end{eqnarray*}
If the last rule in $\pi$ is $E_\linear$, then $\sigma$ is obtained by
using the inductive hypothesis on one of the immediate premises $\phi$ 
of $\pi$ obtaining $\psi$, applying $E_\linear$ to $\psi$ and the
other premise $\xi$ of $\pi$. We have:
\begin{eqnarray*}
S(\sigma,0)&=& S(\psi,0)+S(\xi,0)+1\\
          &\leq& S(\rho,0)\min\{1,S(\phi,0)\}+S(\phi,0)+S(\xi,0)+1\\
           &\leq& S(\rho,0)\min\{1,S(\pi,0)\}+S(\pi,0)\\
\forall i\geq 1.S(\sigma,i)&=& S(\psi,i)+S(\xi,i)\\
           &\leq& S(\rho,i)\min\{1,S(\phi,i)\}+S(\phi,i)+S(\xi,i)\\
           &\leq& S(\rho,i)\min\{1,S(\pi,i)\}+S(\pi,i)
\end{eqnarray*}
If the last rule in $\pi$ is $!$, then $\sigma$ is just obtained from
$\pi$ by weakening lemma, because $x$ cannot appear free in $M$. 
The inequality easily follows.\par
Let us now prove point ii).
If $\pi$ is just an axiom, we can proceed as previously. If the
last rule in $\pi$ is $I^L_\linear$ ($I^I_\linear$), then $\rho$ is
obtained as in point i) and, in both cases:
\begin{eqnarray*}
S(\sigma,0)&=& S(\psi,0)+1\leq S(\rho,0)S(\phi,0)+S(\phi,0)+1\\
         &\leq& S(\rho,0)S(\pi,0)+S(\pi,0)\\
\forall i\geq 1. S(\sigma,i)&=& S(\psi,i)\leq S(\rho,i)S(\phi,0)+S(\phi,i)\\
         &\leq& S(\rho,i)S(\pi,0)+S(\pi,i)
\end{eqnarray*}
If the last rule in $\pi$ is $E_\linear$, then $\sigma$ is
obtained  by using the inductive hypothesis on both the immediate
premises $\phi$ and $\psi$ of $\pi$ obtaining $\xi$ and $\chi$
and applying $E_\linear$ to them. We obtain:
\begin{eqnarray*}
S(\sigma,0)&=& S(\xi,0)+S(\chi,0)+1\\
           &\leq& (S(\rho,0)S(\phi,0)+S(\phi,0))+(S(\rho,0)S(\psi,0)+S(\psi,0))+1\\
           &\leq& S(\rho,0)(S(\phi,0)+S(\psi,0)+1)+(S(\phi,0)+S(\psi,0)+1)\\
           &=& S(\rho,0)S(\pi,0)+S(\pi,0)\\
\forall i\geq 1.S(\sigma,i)&=& S(\xi,i)+S(\chi,i)\\
           &\leq& (S(\rho,i)S(\phi,0)+S(\phi,i))+(S(\rho,i)S(\psi,0)+S(\psi,i))\\
           &=& S(\rho,i)(S(\phi,0)+S(\psi,0)+1)+(S(\phi,i)+S(\psi,i))\\
           &=& S(\rho,i)S(\pi,0)+S(\pi,i)
\end{eqnarray*}
If the last rule in $\pi$ is $!$, the $\sigma$ is again obtained by $\pi$ and the
inequality follows.\par
Let us now prove claim iii). Notice that the last rule in
$\rho$ must be $!$ rule, because
$A$ is modal and $N$ is a value. 
If the last rule in $\pi$ is $I^L_\linear$ ($I^I_\linear$), then $\sigma$ is
obtained in the usual way and:\\
\begin{eqnarray*}
S(\sigma,0)&=& S(\psi,0)+1\leq S(\phi,0)+1 = S(\pi,0)\\
\forall i\geq 1.S(\sigma,i)&=& S(\psi,i)\leq S(\rho,i)(\sum_{j\leq i}S(\phi,j))+S(\phi,i)\\
         &=& S(\phi,i)(\sum_{j\leq i}S(\pi,j))+S(\pi,i)\;\;\forall i\geq 1
\end{eqnarray*}
If the last rule in $\pi$ is $E_\linear$, then we apply the inductive hypothesis 
to the immediate premises $\phi$ and $\psi$ of $\pi$ and to a type derivation which
is structurally equivalent to $\rho$. We obtain $\xi$ and $\chi$ and
apply $E_\linear$ to them, obtaining a type derivation which is structurally
equivalent to the desired $\sigma$. Now we have:
\begin{eqnarray*}
S(\sigma,0)&=& S(\xi,0)+S(\chi,0)+1\leq S(\phi,0)+S(\psi,0)+1=S(\pi,0)\\
\forall i\geq 1.S(\sigma,i)&=& S(\xi,i)+S(\chi,i)\\
           &\leq& S(\rho,i)(\sum_{j\leq i}S(\phi,j))+S(\phi,i)+
                  S(\rho,i)(\sum_{j\leq i}S(\psi,j))+S(\psi,i)\\
           &=& S(\rho,i)(\sum_{j\leq i}(S(\phi,j)+S(\psi,j)))+S(\phi,i)+S(\psi,i)\\
           &=& S(\rho,i)(\sum_{j\leq i}S(\pi,j))+S(\pi,i)\\
\end{eqnarray*}
If the last rule in $\pi$ is $!$, then 
we can  suppose the last rule in $\rho$
to be a $!$ and let $\psi$ be the immediate premise
of $\rho$. We first
apply the induction hypothesis (or one of the
other two claims) to the immediate premise $\phi$ of
$\pi$ and to $\psi$ obtaining $\xi$; then, 
we apply rule $!$ to $\xi$ and we get $\sigma$. 
Clearly, $S(\sigma,0)=0$ by definition. For every
$i\geq 0$, we have that
$$
S(\xi,i)\leq (\sum_{j\leq i}S(\phi,j))S(\psi,i)+S(\phi,i)
$$
independently on the way we get $\xi$. Indeed,
\begin{eqnarray*}
\min\{1,\S(\phi,i)\}S(\psi,i)+S(\phi,i)&\leq&(\sum_{j\leq i}S(\phi,j))S(\psi,i)+S(\phi,i);\\
S(\phi,0)S(\psi,i)+S(\phi,i)&\leq&(\sum_{j\leq i}S(\phi,j))S(\psi,i)+S(\phi,i).
\end{eqnarray*}
As a consequence, for every $i\geq 1$,
\begin{eqnarray*}
S(\sigma,i)&=& S(\xi,i-1)\leq (\sum_{j\leq i-1}S(\phi,j))S(\psi,i-1)+S(\phi,i-1)\\
           &\leq& (\sum_{j\leq i}S(\pi,j))S(\rho,i)+S(\pi,i)
\end{eqnarray*}
This concludes the proof.
\end{proof}

The following can be thought of as a strengthening of subject reduction  
and is a corollary of Lemma~\ref{lemma:substitution2}. 
\begin{prop}\label{prop:subjred2} 
If $\pi:\Gamma\mid\Delta\mid\Theta\vdash M:A$,  
and $M\rightarrow_v N$ by reducing a redex 
at level $i$ in $\pi$, then there is  
$\rho:\Gamma\mid\Delta\mid\Theta\vdash N:A$
such that
\begin{eqnarray*}
\forall j<i.S(\rho,j)&=&S(\pi,j)\\
S(\rho,i)&<&S(\pi,i)\\
\forall j>i.S(\rho,j)&\leq& S(\pi,j)(\sum_{k\leq j}S(\pi,k))
\end{eqnarray*} 
\end{prop}
\begin{proof}
Type derivation $\rho$
is identical to $\pi$ up
to level $i$, so the equality $S(\rho,j)=S(\pi,j)$
holds for all levels $j<i$. At levels
$j\geq i$, the only differences between
$\rho$ and $\pi$ are due to the replacement
of a type derivation $\phi$ for 
$(\lambda x.L)P$ with another type derivation
$\psi$ for
$L\{P/x\}$. $\psi$ is obtained by 
Lemma~\ref{lemma:substitution2}. The needed
inequalities follow from the ones in
Lemma~\ref{lemma:substitution2}.
\end{proof}
If $\pi$ is obtained from $\rho$ by reducing a
redex at level $i$ as in Proposition~\ref{prop:subjred2}, 
then we will write $\pi\rightarrow_v^i\rho$.
Consider now a term $M$ and a derivation
$\pi:\Gamma\mid\Delta\mid\Theta\vdash M:A$.
By Proposition~\ref{prop:subjred2}, every 
reduction sequence 
$M\rightarrow_v N\rightarrow_v L\rightarrow_v\ldots$
can be put in correspondence with a sequence
$\pi\rightarrow_v^i\rho\rightarrow_v^j\sigma\rightarrow_v^k\ldots$
(where $\rho$ types $N$, $\sigma$ types $L$, etc.). The following
result give bounds on the lengths of these sequences and on the
possible growth during normalization.
\begin{prop}\label{prop:bounding}
For every $d\in\N$, there are elementary
functions $f_d,g_d:\N\rightarrow\N$ such that,
for every sequence
$$
\pi_0\rightarrow_v^{i_0}\pi_1\rightarrow_v^{i_1}\pi_2\rightarrow^{i_2}_v\ldots
$$
it holds that
\begin{enumerate}[$\bullet$]
  \item
  For every $i\in\N$, $\sum_{e\leq d}S(\pi_i,e)\leq f_d(S(\pi_0))$;
  \item
  There are at most $g_d(S(\pi_0))$ reduction steps with
  indices $e\leq d$.
\end{enumerate}
\end{prop}
\begin{proof}
We go by induction on $d$ and define
$f_d$ and $g_d$ such that the given
inequalities hold and, additionally,
$f_d(n)\geq n$ for each $n\in\N$. $f_0$ and $g_0$ 
are both the identity on natural numbers,
because $S(\pi_0,0)$ can only decrease during
reduction and it can do that at most
$S(\pi_0,0)$ times. Consider now $d\geq 1$.
Each time $S(\pi_i,d)$ grows, its value
goes from $S(\pi_i,d)$ to at 
most $S(\pi_i,d)(S(\pi_i,d)+f_{d-1}(S(\pi_0)))$, because
by Proposition~\ref{prop:subjred2} it can grow
to, at most $S(\pi_i,d)(\sum_{k\leq d}S(\pi_i,k))$
and, by inductive hypothesis
$$
\sum_{k\leq d}S(\pi_i,k)=S(\pi_i,d)+\sum_{k\leq d-1}S(\pi_i,k)\leq S(\pi_i,d)+f_{d-1}(S(\pi_0)).
$$
We claim that after having increased $n$ times,
$S(\pi_i,d)$ is at most $(f_{d-1}(S(\pi_0))+n)^{2^{n+1}}$.
Indeed, initially 
$$
S(\pi_i,d)\leq S(\pi_0,d)\leq S(\pi_0)\leq (f_{d-1}(S(\pi_0)))^2
$$
And, after $n\geq 1$ increases,
\begin{eqnarray*}
S(\pi_i,d)&\leq& (f_{d-1}(S(\pi_0))+n-1)^{2^{n}}((f_{d-1}(S(\pi_0))+n-1)^{2^{n}}+
                 f_{d-1}(S(\pi_0)))\\    
          &\leq& (f_{d-1}(S(\pi_0))+n)^{2^{n}}((f_{d-1}(S(\pi_0))+n-1)^{2^{n}}+
                 (f_{d-1}(S(\pi_0))+n-1)^{2^{n-1}})\\
          &\leq& (f_{d-1}(S(\pi_0))+n)^{2^{n}}((f_{d-1}(S(\pi_0))+n-1+1)^{2^{n-1}})^2\\
          &=& (f_{d-1}(S(\pi_0))+n)^{2^{n+1}}
\end{eqnarray*}
From the above discussion, it follows that the functions 
\begin{eqnarray*}
f_d(n)&=&f_{d-1}(n)+(f_{d-1}(S(\pi_0))+g_{d-1}(n))^{2^{g_{d-1}(n)+1}}\\
g_d(n)&=&g_{d-1}(n)+\sum_{i=0}^{g_{d-1}(n)}(f_{d-1}(S(\pi_0))+i)^{2^{i+1}}
\end{eqnarray*}
are elementary and satisfy the conditions above.
This concludes the proof.
\end{proof} 
  
\begin{thm}\label{theo:BOUNDS}
For every $d\in\N$ there are elementary functions $p_d,q_d:\N\rightarrow\N$
such that whenever
$\pi:\Gamma\mid\Delta\mid\Theta\vdash M:A$,
the length of call-by-value reduction sequences starting from $M$
is at most $p_{\partial(\pi)}(L(M))$ and the length of any
reduct of $M$ is at most $q_{\partial(\pi)}(L(M))$.
\end{thm} 
\begin{proof}
This is immediate from Proposition~\ref{prop:bounding}.
\end{proof}

%%%%%%%%%%%%%%%%%%%%%%%%
\section{Type Inference}\label{sect:typeinf}
%%%%%%%%%%%%%%%%%%%%%%%%
We prove, in a constructive way, that the type inference problem for
\ETAS{} is decidable. Namely a type inference algorithm is
designed such that, for every lambda term $M$ it produces
a \emph{principal typing} from which all and only its typings can be
obtained by a suitable substitution. The substitution is a partial
function, defined if it satisfies a set of linear constraints.  If
there is not a substitution defined on its principal typing, then $M$
is not typable.  We will also prove that the computational complexity
of the type inference procedure is of the same order as the type
inference for simple type assignment system.

The design of the algorithm is based on the following Generation Lemma.
\begin{lem}[Generation Lemma] 
Let $\Gamma \mid \Delta \mid \Theta \vdash M:A$.
 \begin{enumerate}[\em i)]
\item Let $M\equiv x.$ If $A$ is linear, then either $\{x:A\} \subseteq \Gamma$ or $\{x:A\} 
\subseteq \Theta$. 
Otherwise, $\{x:A\}\in \Delta$.
\item Let $M \equiv \lambda x.N$. Then $A$ is of the shape $\underbrace{!...!}_{n}(B \lin C) 
\ (n \geq 0)   $.
\begin{enumerate}[$\bullet$]
\item Let $n=0$. If $B$ is linear then $\Gamma,x:B \mid \Delta \mid \Theta \vdash N:C$, else
$\Gamma \mid \Delta,x:B \mid \Theta \vdash N:C.$ 
\item Let $n >0$. Then $\emptyset \mid \Delta \mid \emptyset \vdash M:A$ and
$\Gamma' \mid \Delta' \mid \Theta' \vdash N:C$, where 
$\Delta = \underbrace{!...!}_{n}((\Gamma' \cup \Delta' \cup \Theta')-\{x:B\}  )   $.   
\end{enumerate}
\item Let $M \equiv PQ$. Then $A$ is of the shape $\underbrace{!...!}_{n}B  \ (n \geq 0)   $,
$\Gamma_{1}  \mid \Delta' \mid \Theta' \vdash P: C \lin \underbrace{!...!}_{m}B$ and  
$\Gamma_{2}  \mid \Delta' \mid \Theta' \vdash Q:C$, for some $m \leq n$.
\begin{enumerate}[\em(a)]
\item If $n-m=0$, then $\Gamma=\Gamma_{1}\cup \Gamma_{2}  $, $\Delta=\Delta'$ and $\Theta =\Theta'$.
\item If $n-m>0$, then $\emptyset \mid \Delta \mid \emptyset \vdash M:A$, where
$\Delta = \underbrace{!...!}_{n-m}( \Gamma_{1}\cup \Gamma_{2}\cup \Delta' \cup \Theta'  )  $.      
\end{enumerate}
 \end{enumerate}
 \end{lem}  
\noindent The principal typing is described through the notion of a
 \emph{type scheme}, which is an extension of that one used in~\cite{Coppola03tlca} in
 the context of $\LambdaEA$ and \NEAL.  Roughly speaking, a type scheme
 describes a class of types, i.e.  it can be transformed into a type
 through a suitable notion of a substitution.

\begin{defi}\hfill
\begin{enumerate}[i)]
\item\emph{Linear schemes} and \emph{schemes} are
  respectively defined by the grammars
\begin{align*}
  \mu & ::=\alpha \mid \sigma \linear \sigma\\
  \sigma &::= \mu \mid !^p\mu.
\end{align*}
where $\alpha$ belongs to a countable set of scheme variables and the \emph{exponential} $p$  
is defined by the grammar
\begin{displaymath}
 p::= a \mid p+p  
  \end{displaymath}
  where $a$ ranges over a countable set of \emph{literals}.
 Linear schemes are
  ranged over by $\mu, \nu$, schemes are ranged over by $\sigma,
  \tau, \rho$, exponentials are ranged over by $p,q,r$ and literals are ranged over by $a,b$.
  
 Note that the grammar does not generate nesting exponentials, i.e., $!^{p}!^{q}\alpha$ is not a correct scheme, while
 $!^{p+q}\alpha$ is correct. 
\item A \emph{modality set} is a set of linear constraints in the
form either $p=q$ or $p>0$ or $p=0$, where $p$ and $q$ are exponentials. Modality sets are ranged over by C.
\item A \emph{type scheme} is denoted by $\sigma\restriction_{C}$, where $\sigma$ is a scheme 
and $C$ is
a modality set. Type schemes will be ranged over by $\zeta, \theta$. Let $\mathcal{T}$ 
denote the set of type schemes.
  \item $\overline{\sigma \restriction_{C}}$ denotes the simple type skeleton underlying the 
 type scheme $\sigma\restriction_{C}$, and it is defined as follows:
  \begin{align*}
\overline{\alpha\restriction_{C}}=\alpha;\\
\overline{\sigma \lin \tau\restriction_{C}} =\overline{\sigma\restriction_{C}}\lin \overline{\tau\restriction_{C}}\\
 \overline{!^{p}\sigma\restriction_{C}} = \overline{\sigma\restriction_{C}}\\
\end{align*}
\item A \emph{scheme substitution} $S$ is a partial function from type schemes
  to types, replacing scheme variables by types and literals by
  natural numbers, in such a way that constraints in $C$ are satisfied. If $p$ is an exponential,
  let $S(p)$ be the result of applying the scheme substitution $S$ on all the literals in $p$, e.g. 
 if $p$ coincides with $a_{1}+...+a_{n}$, then $S(p)$ is $S(a_1)+...+S(a_n)$. 
 $C$ is satisfied by $S$ if, for every constraint
 $p=q$ ($p>0, p=0$) in $C$, $S(p)=S(q)$ ($S(p)>0, S(p)=0$)
%$S(a_1)+...+S(a_n)=S(b_1)+...+S(b_m)$, and, for every constraint $a_{1}+...+a_{n}>0$, 
%$S(a_1)+...+S(a_n)>0$.
Clearly the solvability of a
set of linear constraints is a decidable problem. 
The application of a substitution 
  $S$ satisfying $C$ to a type scheme $\sigma\restriction_C$ is defined inductively as follows:
  \begin{align*}
    S (\alpha\restriction_C) &=A &\texttt{ if } [\alpha \mapsto A]\in S;\\
    S (\sigma \linear \tau\restriction_C)& =S (\sigma\restriction_C) \linear S (\tau\restriction_C) &;\\
    S (!^{p}\mu\restriction_C)&=\underbrace{!...!}_n S (\mu\restriction_C) & 
    \texttt{ if } p=a_{1}+...+a_{m}\\
    & & n=S(a_1)+...+S(a_m).
  \end{align*}
   If $C$ is not satisfied by $S$, then $S(\sigma\restriction_C)$ is undefined.
% \item Following~\cite{Baillot05tlca} we define a \emph{type assignment} $TA(M)$ for a 
% $\lambda$-term $M$ as a map from variables (free or bound) of $M$ to types.
\end{enumerate}
\end{defi}

Binary relation $\equiv$ is extended to denote the syntactical
identity between both types, schemes and type schemes.  Making clear what we
said before, a type scheme can be seen as a description of the set of
all types that can be obtained from it through a scheme
substitution defined on it. 
% So:
% \begin{varitemize} 
%   \item
%   a linear type scheme describes a set of linear types;
%   \item
%   a modal type scheme describes a subset of modal types; 
%   \item
%   a scheme variable in $\mathit{NVar}$ describes the set of all types.
% \end{varitemize}
\medskip

A substitution is a total function from type schemes to type schemes mapping scheme variables 
to schemes, and generating some constraints. 
A substitution is denoted by a pair 
$\langle s,C\rangle$, where $s$ is a function from scheme variables to schemes and $C$ is a modality set. 
Substitutions will be ranged over by $t$. 
The behaviour of $\langle s,C\rangle$ is defined as follows.
\begin{align*}
    \langle s,C\rangle(\alpha\restriction_{C'}) &=\sigma\restriction_{C \cup C'} &\mbox{    if    } [\alpha \mapsto \sigma]\in s;\\
    \langle s,C\rangle (\sigma \linear \tau\restriction_{C'})& =\sigma' \linear \tau'\restriction_{C''} &\mbox{    if    } 
    \langle s,C\rangle(\sigma\restriction_{C'})=
    \sigma'\restriction_{C'''} \\
    & &\mbox{    and    } \langle s,C'''\rangle(\tau\restriction_{C'})=
    \tau'\restriction_{C''};\\
    \langle s,C\rangle (!^{p}\mu\restriction_{C'})&=!^{p}\nu\restriction_{C''} & \mbox{    if    }\langle s,C\rangle
    (\mu\restriction_{C'})=\nu\restriction_{C''};\\
    \langle s,C\rangle (!^{p}\mu\restriction_{C'})&=!^{r}\nu\restriction_{C''\cup \{r=p+q\}} & \mbox{    if    }
    \langle s,C\rangle(\mu\restriction_{C'})=!^{q}\nu\restriction_{C''} .
  \end{align*}
  Note that the last rule is necessary in order to preserve the 
  correct syntax of schemes, where the nesting of exponentials is not allowed. 

  In order to define the principal typing, we will use a unification
algorithm for type schemes, which is a variant of that defined in
\cite{Coppola03tlca}.  Let $=_e$ be the relation between type schemes
such that $\sigma\restriction_{C} =_{e}\tau\restriction_{C'}$ if $\overline{\sigma\restriction_{C}}\equiv \overline{\tau\restriction_{C}}$.

The unification algorithm, which we will present in Table
\ref{U-unification}, in Structured Operational Semantics style, is a function $U$ from $\mathcal{T}\times\mathcal{T}$
to substitutions.

% following~\cite{Baillot05tlca}, it is easy to show that if $S$ is a
% solution of $C$ then for every strictly positive $n$, $nS$ is a
% solution too, where $nS$ is the scheme substitution such that 
% $nS(\alpha)=S(\alpha)$ and $nS(p)=n \times S(p)$, for every scheme variable 
% $\alpha$ and 
% every literal $p$. Hence the problem of finding a solution for $C$ is
% polynomial in the number of variables and constraints. 

\begin{table*}
\begin{center}
\fbox{
\begin{minipage}{.9\textwidth}
\begin{displaymath}
  \infer[(U1)]{U(\alpha\restriction_C, \alpha\restriction_{C'})=\langle [],\emptyset\rangle }{}
\end{displaymath}
\vspace{1mm}   
\begin{displaymath}
\infer[(U2)]{U(\alpha\restriction_C,\sigma\restriction_{C'})=\langle [\alpha\mapsto\sigma],\emptyset
    \rangle }{\alpha \textrm{
      does not occur in } \sigma }
      \end{displaymath}
% \end{displaymath}
\vspace{1mm}
% \begin{displaymath}
\begin{displaymath}
\infer[(U3)]{U(\sigma\restriction_{C},\alpha\restriction_{C'})=
\langle [\alpha\mapsto\sigma],\emptyset\rangle }{\alpha \text{
      does not occur in } \sigma}
\end{displaymath}
\vspace{1mm}   
\begin{displaymath}
\infer[(U4)]{U(\sigma\linear\tau\restriction_{C},!^{q}\nu\restriction_{C'})=\langle s,C''\cup \{q = 0\}, \rangle}{%
U(\sigma\linear\tau\restriction_{C},\nu\restriction_{C'})=\langle s,C''\rangle}
\end{displaymath}
\vspace{1mm}
\begin{displaymath}
\infer[(U5)]{U(!^{p} \mu\restriction_{C} ,\sigma\linear\tau\restriction_{C'})=
\langle s,C''\cup \{p = 0\}\rangle}{%
U(\mu\restriction_{C}, \sigma\linear\tau\restriction_{C'})=\langle s,C''\rangle}
\end{displaymath}
\vspace{1mm}   
\begin{displaymath}
\infer[(U6)]{U(!^{p} \mu \restriction_{C} ,!^{q}\nu\restriction_{C'} )=
\langle s, C'' \cup \{p =q  \}\rangle }
{U(\mu\restriction_{C} ,\nu\restriction_{C'})=\langle s,C''\rangle }
\end{displaymath}
\vspace{1mm}   
\begin{displaymath}
\infer[(U7)]{U(\sigma_1 \linear \sigma_2\restriction_{C} ,\tau_1 \linear
  \tau_2\restriction_{C'})=\langle s_1,C_{1}\rangle \circ 
  \langle s_2,C_{2}\rangle}
  {\begin{array}{l}
  U(\sigma_1\restriction_{C} ,\tau_1\restriction_{C'} )=
  \langle s_{1},C_1\rangle\\
   U( \langle s_1,C_{1}\rangle(\sigma_2\restriction_{C}) , 
  \langle s_1,C_{1}\rangle(\tau_2\restriction_{C'}))=\langle
  s_2,C_{2}\rangle
  \end{array} }
\end{displaymath}

\vspace{1mm}  
$\langle s_1,C_{1}\rangle \circ 
  \langle s_2,C_{2}\rangle $ is the substitution such that \\
$\langle s_1,C_{1}\rangle \circ 
  \langle s_2,C_{2}\rangle(\sigma\restriction_{C})=\langle s_2,C_{2}\rangle(\langle s_1,C_{1}\rangle
  (\sigma\restriction_{C}))$.
\end{minipage}}
\end{center}
\caption{The unification algorithm $U$}\label{U-unification}
\end{table*}

The following lemma  proves that the function $U$ supplies a more general unifier for two type schemes,
in two steps: the substitution it generates is the most general unifier with respect to
the relation $=_{e}$, and moreover there is a most general scheme substitution
unifying the two type schemes modulo the syntactic equivalence $\equiv$.

\begin{lem}\label{e-unif}\hfill
\begin{enumerate}[\em i)]
\item (correctness) $U(\sigma\restriction_{C}, \tau\restriction_{C'})=\langle s,C''\rangle $ implies
  $\langle s,C''\rangle(\sigma\restriction_{C})=\sigma'\restriction_{C'''}$,   
  $\langle s,C''\rangle(\tau\restriction_{C'})=\tau'\restriction_{C''''}$ where $\sigma'\restriction_{C'''}=_{e}\tau'\restriction_{C''''}$.
  Moreover for every scheme substitution $S$, defined on both the type schemes,
  $S(\sigma'\restriction_{C'''})\equiv  S(\tau'\restriction_{C''''})$.
\item (completeness) $S(\sigma\restriction_{C})\equiv S(\tau\restriction_{C'})$ implies
  $U(\sigma\restriction_{C},\tau\restriction_{C'})=\langle s,C''\rangle $ 
%  and $S(\sigma\restriction_{C})=S'(\langle s,C''\rangle(\sigma\restriction_{C}))$
  and  there is $S'$ such that, for every type scheme $\zeta$, 
  $S(\zeta)=S'(\langle s,C''\rangle(\zeta))$, for some $S'$.
\end{enumerate}
\end{lem}

\proof\hfill
\begin{enumerate}[i)]
\item Easy, by induction on the rules defining $U(\sigma\restriction_{C}, \tau\restriction_{C'})$.
\item By induction on the pair ($n,m$), where $n$ and $m$ are respectively the number of scheme variables occurring in
both $\sigma$ and $\tau$ and the total number of symbols of $\sigma$ and $\tau$.\\
Let $\sigma \equiv \alpha$, and let $S(\sigma\restriction_{C})\equiv S(\tau\restriction_{C'})\equiv A$; clearly 
either $\tau \equiv \alpha$ or $\alpha$ cannot occur as proper subterm of $\tau$. 
In the first case $U(\alpha\restriction_{C},\alpha\restriction_{C'})=\langle [], \emptyset\rangle $, and
$S =S'$. In the second case
$U(\sigma\restriction_{C},\tau\restriction_{C'})=\langle [\alpha \mapsto \tau], \emptyset\rangle $. Then every 
scheme substitution $S'$, solving both $C$ and $C'$ and acting as $S$ on all the scheme 
variables occurring in $\sigma$ and $\tau$ but $\alpha$
does the desired job.\\
Let $\sigma \equiv \sigma_{1}\linear \sigma_{2}$ and $\tau \equiv \tau_{1}\linear \tau_{2}$.
So $S(\sigma_{1}\linear \sigma_{2}\restriction_{C})\equiv S(\sigma_{1}\restriction_{C})\linear S(\sigma_{2}\restriction_{C})$,
and $S(\tau_{1}\linear \tau_{2}\restriction_{C'})\equiv S(\tau_{1}\restriction_{C'})\linear S(\tau_{2}\restriction_{C'})$.
So by induction 
$U(\sigma_{1}\restriction_{C},\tau_{1}\restriction_{C'})=\langle s_{1}, C_{1}\rangle $,
$\langle s_{1}, C_{1}\rangle (\sigma_{1}\restriction_{C})= \sigma'\restriction _{C'_{1}}$
and $\langle s_{1}, C_{1}\rangle (\tau_{1}\restriction_{C'})= \tau'\restriction _{C''_{1}}$
where $\sigma'\restriction _{C'_{1}}=_{e}\tau'\restriction _{C''_{1}}$. 
Moreover there is $S_{1}$
such that $S(\sigma_{2}\restriction_{C}) \equiv S_{1}(\langle s_{1}, C_{1}\rangle
(\sigma_{2}\restriction _{C}))$ and
$S(\tau_{2}\restriction_{C'})\equiv S_{1}(\langle s_{1}, C_{1}\rangle
(\tau_{2}\restriction_{C'}))$.
In case $s_{1}$ is not empty, the number 
of scheme variables in both the type schemes is less than
in $\sigma $ and $\tau$; otherwise their total number of symbols is less than the one in 
$\sigma$ and $\tau$.
In both cases we can apply the induction hypothesis and conclude the proof.\\
All the other cases follow directly from the induction hypothesis.\qed
\end{enumerate}

The principal type scheme of a term is a pair in the form
$\langle\Sigma;\zeta\rangle$, where $\Sigma$ is a scheme context (i.e., a set of assignments
of the shape $x:\theta$, where no variable is repeated), and $\zeta$ is a
type scheme. 

In order to simplify the text of the algorithm, we will use the following conventions:
\begin{enumerate}[$\bullet$]
\item Let $\sigma $ be a scheme. $!^{p}\sigma $ denotes $!^{p}\mu$ in case $\sigma \equiv \mu$,
$!^{r}\mu$, where $r=p+q$ in case $\sigma \equiv !^{q}\mu$; if $\zeta \equiv \sigma \restriction_{C}$, then
$!^{p}\zeta $ denotes $!^{p}\sigma \restriction_{C}$;
\item Let $\Sigma $ be a scheme context. $!^{p}\Sigma$ denotes the scheme context 
$\{x:!^{p}\zeta \mid x:\zeta \in \Sigma \}. $
\end{enumerate}

The principal type scheme algorithm is defined in
Table~\ref{table-def-type-inf}.

\begin{table*}
\begin{center}
\scalebox{.9}{%
\framebox{%
  \begin{minipage}{.9\textwidth}
\begin{itemize}
\item $ \PT(x)=\langle\{x:!^a \alpha\restriction_{\emptyset}\},
!^a \alpha\restriction_{\emptyset}\rangle$
\item \begin{tabbing}
    $\PT(\lambda x.M)=$ \= \texttt{let} \=
    $\PT(M)=\langle\Sigma,\sigma\restriction_{C}\rangle$ \texttt{in}\\ 
    \>\> \texttt{if} \=$x$ doesn't occur free in $M$\\
     \>\>$\langle!^b\Sigma, !^b(!^{a}\alpha\linear\sigma)\restriction_{C}\rangle$\\
    \>\> \texttt{else}\\
    \>\>\> \texttt{let} \=$\Sigma = \{x:\tau\restriction_{C'}\}\cup\Sigma'$ \texttt{in}\\
    \>\>\>\> \texttt{if} \= $x$ occurs multiple times in $M$\\
    \>\>\>\>\> \texttt{let} \=$U(\tau\restriction_{C'},!^{a}\alpha\restriction_{\{a>0\}})=\langle s_{1},C_1\rangle$ 
    \texttt{in}\\
  %  \>\>\>\>\>\> \texttt{let} \= $F(s_1,\Sigma')=\langle\Sigma'',C_2\rangle$ \texttt{and}\\
  %  \>\>\>\>\>\>\> $F(s_1,!^{p_1}\alpha\linear
 %   \sigma)=\langle\mu,C_3\rangle$,\\
 %   \>\>\>\>\>\>\texttt{in let} \=$F(p,\Sigma'')=\langle\Sigma''',C_4\rangle$ \texttt{in}\\
 \>\>\>\>\>\> \texttt{let} \=$\langle s_{1},C_1\rangle(\tau\restriction_{C'})=\tau'\restriction_{C''} $  \texttt{and}\\
  \>\>\>\>\>\>\>  $\langle s_{1},C_1\rangle(\sigma\restriction_{C})=\sigma'\restriction_{C'''} $ 
  \texttt{in}\\
    \>\>\>\>\>\>\> $\langle \langle s_{1},C_1\rangle(!^{b}\Sigma'), 
    !^{b}(\tau' \lin \sigma')\restriction_{C'' \cup C'''}\rangle$\\
    \>\>\>\> \texttt{else}\\
    \>\>\>\>\> $\langle !^{a}\Sigma', !^a(\tau\linear \sigma)\restriction_{C\cup C'} \rangle$
  \end{tabbing}
\item \begin{tabbing}
   $\PT(M_1\ M_2)=$ \= \texttt{let} \=
     $\PT(M_1) = \langle \Sigma_1, \sigma_1\restriction_{C_1} \rangle$ \texttt{,}
     $\PT(M_2) = \langle \Sigma_2, \sigma_2\restriction_{C_2}\rangle$ \\
 \>\>   \texttt{and let they are disjoint in}\\
 \>\> \texttt{let} \= $\{x_1,\ldots,x_k\} = dom(\Sigma_1)\cap
dom(\Sigma_2)$ \texttt{and} \\
\>\>\> $\forall 1\le i\le k\ \forall 1\le j\le 2\
x_i:\zeta_j^i\in\Sigma_j$ \texttt{and}\\
\>\>\> $U(\sigma_1\restriction_{C_1},(\sigma_2\linear!^{a}\alpha)\restriction_{C_2}) = \langle s',C'\rangle$
\texttt{and}\\
\>\>\> $U(\zeta_1^1,\zeta_2^1) = t_{1}$ \texttt{and}\\
\>\>\> $U(t_{i-1}(\zeta_1^i),t_{i-1}(\zeta_2^i)) = t_{i}$ \texttt{and} \\
\>\>\>\texttt{let} $t=t_{1}\circ t_{2}\circ...\circ t_{k}$ \texttt{in}\\
%\texttt{let} $tt=t_{1}\circ t_{2}\circ...\t_{k}$  \texttt{in}\\
\>\>\>$\langle !^{b}(t(\Sigma_{1})\cup t(\Sigma_{2})), !^{b}(t(!^{a}\alpha)\restriction_{C'})\rangle$
  \end{tabbing}
\end{itemize}
$\alpha, a,b$ are fresh variables.
\end{minipage}
}}
\caption{The type inference algorithm $\mathit{PT}$.} \label{table-def-type-inf}
\end{center}
\end{table*}

%The following technical property will be useful for proving the Type Inference Theorem.

%\begin{prty}\label{exp}
%Let $PT(M)=\langle \Sigma, \sigma\restriction_{C} \rangle$. 
%\begin{enumerate}
%\item $\sigma$ is of the shape $!^{p}\mu $, for some $p$;
%\item if $\Sigma = \{x_{i}:\tau_{i}\restriction_{C_{i}}\mid i \in I \}$, for some $I$, then 
%$\tau_{i}\equiv !^{p_{i}}\nu_{i}$, for some $p_{i}$;
%\item $C$ and $C_{i}$ do not contain constraints on $p$ and $p_{i}$ ($i \in I$) . 
%\end{enumerate}
%\end{prty} 
%\begin{proof}
%Easy, by inspecting the algorithm in Table \ref{table-def-type-inf}.
%\end{proof}

\begin{thm}[Type Inference]\hfill 
\begin{enumerate}[\em i)] 
\item (correctness) If $\PT(M) = \langle \Sigma, \zeta\rangle$ then
for every scheme substitution
$S$ defined on all the type schemes occurring in $PT(M)$, $\exists \Gamma,\Delta,\Theta$ partitioning
$S(\Sigma)$ s.t. $\Gamma \mid \Delta \mid \Theta \vdash M: S(\zeta)$.
\item (completeness)  If $\Gamma \mid \Delta \mid \Theta \vdash M: A$ then
 $\PT(M)=\langle \Sigma,\zeta\rangle $ and there exists a scheme
 substitution $S$ defined on all the type schemes occurring in $PT(M)$ such that 
 $S(\Sigma) \subseteq \Gamma \cup \Delta \cup \Theta $ and $A=S(\zeta)$.
\end{enumerate} 
\end{thm}  

\proof\hfill
\begin{enumerate}[i)] 
\item By induction on $M$.
\begin{enumerate}[$\bullet$]
\item $M \equiv x$. Then $PT(M)=\langle\{x:!^a \alpha\restriction_{\emptyset}\};
!^a \alpha\restriction_{\emptyset}\rangle$.
Every scheme substitution $S$ satisfies the empty set of constraints. 

If 
$S(!^a \alpha\restriction_{\emptyset})$ is linear, then
take $\Delta=\emptyset$ and either $\Gamma =\{x:S(!^a \alpha\restriction_{\emptyset})\}$ 
and $\Theta = \emptyset$ 
or $\Theta = \{x:S(!^a \alpha\restriction_{\emptyset})\}$ and $\Gamma =\emptyset$. 
Otherwise choose 
$\Gamma =\Theta =\emptyset$, and $\Delta=\{x:S(!^a \alpha\restriction_{\emptyset})\}$.
\item $M \equiv \lambda x.P$. This case follows directly by induction.
\item $M\equiv M_{1}M_{2}$. 
Then $\PT(M_1) = \langle \Sigma_1, \theta_{1}\rangle$, 
$\PT(M_2) = \langle \Sigma_2, \theta_{2}\rangle$, and 
$PT(M)=\langle \Sigma',\theta' \rangle $,
where $\theta' $ is defined as in Table \ref{table-def-type-inf}.

Let $S$ be a scheme substitution defined on all the type schemes occurring in $PT(M)$:
note that, by the definition of the function $U$ this implies that $S$ is defined on both $PT(M_{1})$
and $\PT(M_2)$. Moreover, by construction of $PT$, $x:\zeta_{i} \in \Sigma_{i}$ ($1 \leq i \leq 2$)
implies $U(\zeta_{1},\zeta_{2})$ is defined, and $S(\zeta_{1})\equiv S(\zeta_{2})$,
by Lemma \ref{e-unif}.\\
By induction 
$\Gamma_{i} \mid \Delta_{i} \mid \Theta_{i}\vdash M_{i}:S(\theta_{i})$
($1 \leq i \leq 2$). Note that every type derivation for $M$ ends with an application of the rule $(E^{\lin})$, followed by a sequence, 
may be be empty, of rule ($!$). Since in rule $(E^{\lin})$ the two linear contexts are disjoint, we can build the partition of the contexts in
the following way:\\
$\Gamma_{1}=\{x:S(\zeta) \mid (\{x:S(\zeta)\}\subseteq S(\Sigma_{1})) \wedge
(x\in FV(M_{1}) \wedge x\not \in FV(M_{2})) \wedge
(S(\zeta) \mbox{ is linear })\}$,\\
$\Gamma_{2}=\{x:S(\zeta) \mid \{x:S(\zeta)\}\subseteq S(\Sigma_{2}) 
\wedge (x\in FV(M_{2}) \wedge x\not \in FV(M_{1})) \wedge
(S(\zeta) \mbox{ is linear })\}$,\\
$\Delta_{i}=\{x:S(\zeta) \mid (\{x:S(\zeta)\}\subseteq S(\Sigma_{i}))  \wedge
(S(\zeta) \mbox{ is modal })\}$ and\\
$\Theta_{i}=\Sigma_{i}-(\Gamma_{i} \cup \Delta_{i})  $ ($1 \leq i \leq 2$).\\
 By the weakening Lemma, we have that 
$\Gamma_{i} \mid \Delta_{1},\Delta_{2} \mid \Theta_{1},\Theta_{2}
\vdash M_{i}:S(\theta_{i})$ ($1 \leq i \leq 2$). 
Since $S(\theta_{1})\equiv S(\theta_{2})\lin
S(!^{a}\alpha\restriction_{C_{2}}))$ (by Lemma \ref{e-unif}.i), 
the proof follows by rule ($E_{\lin}$).
\end{enumerate}
\item By induction on the derivation proving $\Gamma \mid \Delta \mid \Theta \vdash M: A$.

Let the last used rule be ($A^{L}$). Then $M \equiv x$, and $\{x:A\} \subseteq\Gamma $. By definition, 
$ \PT(x)=\langle\{x:!^a \alpha\restriction_{\emptyset}\};!^a \alpha\restriction_{\emptyset}
\rangle$, and the proof follows easily.

The case ($A^{P}$) is similar.

The cases $(I^{L}_{\lin})$ and $(I^{I}_{\lin})$ both follow by induction and weakening lemma.

Let us consider the case when the last used rule is $(E_{\lin})$. 
Then $M \equiv M_{1}M_{2}$, and 
$\Gamma_{1}\mid \Delta \mid \Theta \vdash M_{1}:B \lin A $ and 
$\Gamma_{2}\mid \Delta \mid \Theta \vdash M_{2}:B $, for some $B$,
where $\Gamma_{1}$ and $\Gamma_{2}$ are disjoint.
By induction $\PT(M_i) = \langle \Sigma_i, \zeta_{i}\rangle$ and 
there is $S_{i}$ defined on all type schemes in $PT(M_i)$ such that $S_{1}(\zeta_{1})\equiv B\lin A$,
and $S_{2}(\zeta_{2})\equiv B $.\\
Moreover 
$S_{i}(\Sigma_{i})\subseteq \Gamma_{i}\cup \Delta \cup \Theta$ ($1 \leq i \leq 2$). 
Since by construction $PT(M_{1})$ and $PT(M_{2})$ are disjoint, there is a well defined scheme substitution
$S$ acting as both $S_{1}$ and $S_{2}$ and such that 
$S(!^a\alpha\restriction_{\emptyset})\equiv A$, for $\alpha, a$ fresh. So
if $\zeta_{2}\equiv \sigma\restriction_{C}$,
$S(\zeta_{1})\equiv S(\sigma\lin !^{p_{1}}\alpha\restriction_{C})$, 
and $S$ is defined on all the type schemes in
$PT(M_i)$ ($1\leq i \leq 2$).
%Let $S'$ be the substitution defined on all the type schemes in $PT(M_{1})$  and $PT(M_{2})$,
%and
%such that $S'(\beta)\equiv S(\beta)$ and $S'(q)=S(q)$, 
%for all scheme variable $\beta$ but a fresh scheme variable $\alpha$, and for all literal $q$ except
%a fresh literal $p_{1}$ where $S'(\alpha)\equiv B$ and $S'(p_{1})=0$.
Then 
%$S'(\zeta_{1})\equiv S'(\sigma\lin !^{p_{1}}\alpha)$, and,
by Lemma \ref{e-unif}.ii), 
$U(\zeta_{1},\sigma\lin !^{p_{1}}\alpha\restriction_{C})=\langle s',C' \rangle$.
Since $S$ satisfies all the constraints in $PT(M_{i})$ ($1\leq i \leq 2$), if 
$x:\theta_{1}$ and $x:\theta_{2}$ belong to $\Sigma_{1}$ and $\Sigma_{2}$
respectively, then $S(\theta_{1})\equiv S(\theta_{2})$, and so, by Lemma \ref{e-unif}.ii), 
they can be unified.
So $PT(M_{1}M_{2})$ is defined, and by 
induction it enjoys the desired properties.

Let the last used rule be $(!)$. Then $A \equiv !A'$ and the premise of the rule is 
$\Gamma' \mid \Delta' \mid \Theta' \vdash M:A'$, where 
$!\Gamma' \cup !\Delta' \cup ! \Theta' \subseteq \Delta $.
By induction $PT(M)=\langle \Sigma, \zeta \rangle$ and there is a scheme substitution $S'$
such that $S'(\Sigma)\subseteq \Gamma' \cup \Delta' \cup \Theta' $.
Let $S$ be such that $S$ is defined on all the type schemes in $PT(M)$, and
such that $S(\zeta)\equiv !S'(\zeta)$, for all 
type scheme $\zeta$: it is easy to check that, if $S'$ is defined, i.e., it satisfies all the constraints in $PT(M)$, than $S$ is 
well defined too, and so it does the right job.\qed
\end{enumerate}

%\begin{lem}\label{lem:nS}
%Let $nS$ be the substitution defined by $S$ such that $nS(\alpha)=S(\alpha)$ and $nS(p)=n \times S(p)$ respectively for every scheme variable $\alpha$ and  every literal $p$.
%\begin{enumerate}
%\item Let $C$ be a set of linear constraints produced by the unification algorithm $U$. Then for every substitution $S$ satisfying $C$, and for every positive natural number $n$, the substitution $nS$ satisfies $C$ too.
%\item Let $C$ be a set of linear constraints produced by the type inference algorithm $PT$. Then for every substitution $S$ satisfying $C$, and for every positive natural number $n$, the substitution $nS$ satisfies $C$ too.
%\end{enumerate}
%\end{lem}
%\begin{proof}
%The proof trivially holds by noticing that $U$ and $F$ produce equalities and
%the set of constraints generated by $PT$ is a union of set produced by $U$ and $F$ plus some inequalities imposing some variables to be greater than or equal to 1.
%\end{proof}

\noindent The complexity of the type inference algorithm $PT$ is of
the same order as the type inference algorithm for simple types. In
order to prove this, we need some notations.  If $A(n)$ is an
algorithm running on a datum $n$, let us denote by $|A(n)|$ its
asymptotic complexity.  Moreover, if $\sigma$ is a scheme, let
$|\sigma | $ be the number of symbols in it.  Let $\mathit{TA}(M)$ be
the type inference algorithm for simple types running on a term
$M$. By abuse of notation, we assume that, for every type scheme
$\sigma$, $\overline{\sigma}$ denotes a simple type: in fact the
syntax of type schemes, when erasing exponentials and constraints,
coincides with that of simple types.
\begin{thm}[Complexity]
$PT(M) \in \mathit{O}(L(M)+|\mathit{TA}(M)|)$.
\end{thm}
\proof
First of all, let us observe that the the unification algorithm $U$ coincide with the Robinson unification, when it is applied to two type schemes
whose set of constraints is empty, so, if $RU$ denotes the Robinson's unification, 
$|U(\sigma \restriction_{\emptyset},
\tau \restriction_{\emptyset} )|=|RU (\overline{\sigma},\overline{\tau})|$. 
Then 
$|U( \sigma\restriction_{C} , 
  \tau\restriction_{C'})|\leq |RU(\overline{\sigma},\overline{\tau})| +|\sigma | + |\tau |$.
Remember that Robinson unification is equivalent to the principal simple type assignment.
Moreover it is easy to see that, if $PT(M)=\langle \Sigma, \sigma \restriction_{C} \rangle$, then 
$\mathit{TA}(M)=\langle\overline{\Sigma}, \overline{\sigma}\rangle$, when $\overline{\Sigma}$ denotes the context obtained 
from $\Sigma$ by applying the function $\overline{[.]}$ to all types in it.

The proof is by induction on $M$. If $M $ is a constant the proof is trivial.
If $M \equiv \lambda x. N$, then $|PT(M)|= |PT(N)| + |U(\tau\restriction_{C}, !^{a}\alpha \restriction_{a>0})| =
|PT(N)| + k$, for a constant $k$, since $\alpha$ is a scheme variable. Then the result follows by induction.
Let $M\equiv PQ$. Then $|PT(PQ)| = |PT(P)| + |PT(Q)| + |U(\sigma\restriction_{C}, 
  (\tau\linear !^{a}\alpha)\restriction_{C'})| + |U(\Sigma_{1},\Sigma_{2})|$ 
if $PT(P)=\langle \Sigma_{1},\sigma\restriction_{C} \rangle$ and 
  $PT(Q)=\langle \Sigma_{2},\tau\restriction_{C'}\rangle $. $U(\Sigma_{1},\Sigma_{2})$ is an abbreviation for the unification of the 
  two scheme contexts $\Sigma_{1}$ and $\Sigma_{2}$, as specified in the algorithm. 
 By induction $PT(PQ)= |\mathit{TA}(P)| +k_{1} + |\mathit{TA}(Q)| + k_{2} + |RU(\overline{\sigma},\overline{\tau})| +|\sigma| +|\tau|+
 |RU(\overline{\Sigma_{1}},\overline{\Sigma_{2}}) + |\Sigma_{1}| + |\Sigma_{2}|$.
 Remembering that $|\mathit{TA}(PQ)|= |\mathit{TA}(P)| + |\mathit{TA}(Q)| + |RU(\overline{\sigma},\overline{\tau})| +
 |RU(\overline{\Sigma_{1}},\overline{\Sigma_{2}})|$, the result follows.\qed
 
%More precisely, the the running time of $PT(M)$ is the maximum between the running time of $PT(x)$, $PT(\lambda x.M)$ and $PT(M_1\ M_2)$. 
%The running time of $PT(x)$ is a constant. 
%The running time of $PT(\lambda x.M)$ is a constant plus the running time of the occur check, that is linear in $L(M)$, 
%plus the running time of the unification algorithm and the running time of $PT(M)$. The running time of $PT(M_1\ M_2)$ is 
%equivalent to the running time of the unification algorithm. 
%Moreover, if $|TA(M)|=|RU(\overline{\sigma},|\overline{\tau})$, then 
%$|PT(M)|=|U(\sigma\restriction_{C}, \tau\restriction_{C'})| + $, then 

%For the same reasons the unification algorithm is bounded by $TA(M)$.
%$|PT(M)| = |U(\sigma)$
%Similarly the running time of $PT(M_1\ M_2)$ is bounded by a polynomial in the size of $TA(M_1\ M_2)$.

%Finally notice that the number of constraints is polynomially bounded by the size of $TA(M)$ too and then, by Lemma~\ref{lem:nS} and following the same technique of~\cite{Baillot05tlca}, we can solve the set of constraints in $\mathbb{Q}$ in polynomial type and then obtain an integer solution.

%%% Local Variables: 
%%% mode: latex
%%% TeX-master: "llcbv"
%%% End: 
\begin{exa}
  Let us illustrate the application of the type inference algorithm to
  the term:
$$\underline{2}\mbox{ }\underline{3}\equiv (\lambda xy. x(xy))\lambda
xy.x(x(xy)).
$$ 
  Let $a,b,c,d,e,f,g,h,k,i,m,n,p,q,r$ be literals and $\alpha,
  \beta,\gamma, \delta,\epsilon$ be scheme variables.  First we will
  build PT(\underline{2}). Starting from
$$PT(x)=\langle \{x:!^{a}\alpha\restriction_{\emptyset}\},!^{a}\alpha\restriction_{\emptyset} \rangle 
\mbox{ and } PT(y)=\langle
\{y:!^{b}\beta\restriction_{\emptyset}\},!^{b}\beta\restriction_{\emptyset}
\rangle,
$$
  due to $U(!^{a}\alpha\restriction_{\emptyset},!^{b}\beta \linear
  !^{c}\gamma\restriction_{\emptyset})=\langle [\alpha \mapsto
  !^{b}\beta \linear !^{c}\gamma ], \{a=0\} \rangle$, we obtain:
$$PT(xy)=\langle \{x:!^{d+a}(!^{b}\beta \linear !^{c}\gamma)\restriction_{C_{0}}, y:!^{d+b}\beta\restriction_{C_{0}}\}, 
!^{c+d}\gamma\restriction_{C_{0}}\rangle
$$ 
  where $C_{0}=\{a=0\}$.

  Now a fresh version of $PT(x)$ is needed, so consider $\langle
  \{x:!^{a_{1}}\alpha_{1}\restriction_{\emptyset}\},!^{a_{1}}\alpha_{1}
  \restriction_{\emptyset} \rangle$.  The rule for application allows
  us to perform certain unifications.  First we obtain
$$U(!^{a_{1}}\alpha_{1}\restriction_{\emptyset}, (!^{c+d}\gamma \linear !^{e} \delta)\restriction_{C_{0}})=
\langle [\alpha_{1}\mapsto !^{c+d}\gamma \linear !^{e} \delta],
C_{1}\rangle
$$ 
  where $C_{1}=\{a_{1}=0 \}$.  A second unification is
  necessary for unifying the two premises on $x$ in the first
  component of $PT(xy)$ and $PT(x)$, respectively:
$$U(!^{a_{1}}\alpha_{1}\restriction_{\emptyset}, !^{d+a}(!^{b}\beta \linear !^{c}\gamma)\restriction_{C_{0}})=
\langle [\alpha_{1}\mapsto !^{b}\beta \linear !^{c}\gamma], C_{2}
\rangle
$$ 
  where $C_{2}= \{a_{1}=a+d\} $.  By composing the two substitutions,
  we have $\langle [\gamma \mapsto \beta, \gamma \mapsto \delta],
  \{c+d=b, e=c \}\rangle$.  So 
$$PT(x(xy))=\langle
  \{x:!^{h+a_{1}}(!^{b }\beta \linear !^{e}\beta)\restriction_{C_{3}},
  y:!^{h+b+d}\beta\restriction_{C_{3}}\}, !^{h+e }\beta
  \restriction_{C_{3}} \rangle
$$
  where $C_{3}=C_{0}\cup C_{1}\cup C_{2}\cup \{c+d=b, c=e\}.$ By
  applying the rule for the abstraction, we obtain:
$$PT(\lambda y.x(xy))=\langle \{x:!^{f}(!^{b }\beta \linear !^{e}\beta)\restriction_{C_{3}}\}, !^{k}(!^{h+b+d}\beta \linear
!^{h+e} \beta)\restriction_{C_{3}}   \rangle
$$
  and
$$PT(\lambda xy.x(xy)) = \langle\emptyset; !^{i}(!^{f}(!^{b}\beta \linear !^{e}\beta)\linear !^{k}(!^{h+b+d}\beta \linear
!^{h+e} \beta))\restriction_{C_{4}}   \rangle
$$
  where $C_{4}=C_{3}\cup \{f=k+h+a_{1}, f>0\}$.

  Due to the particular form of $PT(\lambda xy.x(xy))$, we can deduce
  that the term $\lambda xy.x(xy)$ can be assigned, among others,
  the following types
$$!(A \lin A)\lin !(A\lin A), !(A\lin A)\lin !A \lin !A, !!(A\lin
A)\lin !(!A \lin !A).
$$
  In particular, the scheme substitution that replaces $\beta$ by $!(A
  \lin A)$, furthermore $b$, $e$, $h$, $d$ and $a_1$ by $0$, and both
  $k$ and $f$ by $1$, satisfies the constraints and generates the
  typing $\emptyset \mid \emptyset \mid \emptyset \vdash
  \underline{2}:!(A \lin A) \lin !(A \lin A)$, whose derivation is
  shown in Example \ref{exa:terms}.

  In order to build the principal type scheme of $\underline{3}$, we
  need to start from two fresh copies of $PT(x)$ and $PT(x(xy))$, let
$$\langle \{x:!^{n} \epsilon \restriction_{\emptyset}\}, !^{n} \epsilon \restriction_{\emptyset}\rangle \mbox{   and   }
\langle \{x:!^{h'+a'_{1}}(!^{b'}\alpha \linear !^{e'}\alpha)\restriction_{C'_{3}}, y:!^{h'+b'+d'}\alpha\restriction_{C'_{3}}\},
!^{h'+e' }\alpha \restriction_{C'_{3}} \rangle
$$
  where $C'_{3}=\{a'=0,a'_{1}=0,a'_{1}=d'+a', b'=c'+d', e'=c' \}$.

  By applying the rule for application, we obtain, for the first
  unification:
$$U(!^{n} \epsilon\restriction_{\emptyset},!^{h'+e' }\alpha \linear !^{p}\eta\restriction_{C'_{3}} )= 
\langle [\epsilon \mapsto !^{h'+e' }\alpha \linear !^{p}\eta ],
C_{5}\rangle
$$
  where $C_{5}=\{n=0\}.$ By unifying the two premises on $x$, we have
$$U(!^{n} \epsilon\restriction_{\emptyset},
!^{h'+a'_{1}}(!^{b '}\alpha \linear !^{e'}\alpha)\restriction_{C'_{3}} )=\langle [\epsilon \mapsto !^{b '}\alpha \linear !^{e'}\alpha] 
, C_{6} \rangle
$$ 
  where $C_{6 }= \{n=h'+a'_{1}\}$. So, composing the two substitutions:
$$PT(x(x(xy)))=\langle \{x:!^{q+h'+a'_{1}}(!^{b'}\alpha \linear !^{e'}\alpha)\restriction_{C_{7}}, 
y:!^{q+h'+b'+d'} \alpha\restriction_{C_{7}} \},
!^{p+q}\alpha\restriction_{C_{7}} \rangle$$
  where $C_{7}=C'_{3}\cup C_{5}\cup C_{6}\cup \{b'=h'+e',e'=p  \}$.
  So, by applying the rules for abstraction, we have:
$$PT(\lambda y.x(x(xy)) )=\langle \{x:!^{r+q+h'+a'_{1}}(!^{b'}\alpha
\linear !^{e'}\alpha)\restriction_{C_{7}} \}, !^{r}(!^{q+h'+b'+d'}
\alpha \linear !^{p+q}\alpha)\restriction_{C_{7}}
\rangle
$$
  and
$$PT(\lambda x y.x(x(xy)) )=\langle \emptyset,!^{s}(!^{g}(!^{b'}\alpha
\linear !^{e'}\alpha) \linear !^{r}(!^{q+h'+b'+d'} \alpha \linear
!^{p+q}\alpha))\restriction_{C_{8}} \rangle
$$
  where $C_{8}=C_{7}\cup \{g=r+q+h'+a'_{1}, g>0   \}$.

  It is easy, but boring, to check that the typings for
  $\underline{3}$ are the same that the ones for $\underline{2}$, by
  inspecting the modality set.  Now, in order to build
  $PT(\underline{2}\mbox{ }\underline{3})$, we need to unify the two
  type schemes:
$$\sigma \equiv !^{i}(!^{k+h+a_{1}}(!^{b }\beta \linear
!^{e}\beta)\linear !^{k}(!^{h+b+d}\beta \linear !^{h+e}
\beta))\restriction_{C_{4}}
$$
  and 
$$\tau \equiv !^{s}(!^{g}(!^{b'}\alpha \linear !^{e'}\alpha) \linear
!^{r}(!^{q+h'+b'+d'} \alpha \linear !^{p+q}\alpha))\linear
!^{t}\gamma)\restriction_{C_{8}}
$$ 
  obtaining:
$$U(\sigma, \tau)=\langle [\beta \mapsto !^{b'}\alpha\linear !^{e'}\alpha, \gamma \mapsto !^{h+b+d}\beta \linear !^{h+e}\beta] , C_{9}  
\rangle
$$
  where $C_{9}=\{ i=0, k+h+a_{1}=s,t=k,s=k+h+a_{1},b=g, e=r,
  b'=q+h'+b'+d',e=p+q \}$. So
$$PT(\underline{2}\mbox{ }\underline{3})=\langle \emptyset, !^{t}(!^{h+b+d}(!^{b'}\alpha \linear !^{e'}\alpha)\linear !^{h+e}(!^{b'}\alpha
\linear !^{e'}\alpha))\restriction_{C_{10}}  \rangle
$$
  where $C_{10}= C_{8}\cup C_{9}$.

  Finally, the scheme substitution that replaces $\alpha$ by $A$,
  furthermore $b$, $e$, $b'$ and $e'$ by $0$, and both $t$ and
  $h$ by $1$, satisfies the constraints and generates the typing
  $\emptyset \mid \emptyset \mid \emptyset \vdash\underline{2}\mbox{ }
  \underline{3}:!(!(A\lin A)\lin!(A \lin A))$, whose derivation is
  shown in Example \ref{exa:terms}.
\end{exa}

%%%%%%%%%%%%%%%%%%%%%%%%
\section{Achieving completeness}\label{sect:completeness}
%%%%%%%%%%%%%%%%%%%%%%%%
The type-system we introduced in this paper is not
complete for the class of elementary time functions,
at least if we restrict to uniform encodings.
Indeed, simply typed lambda-calculus \emph{without constants}
is itself incomplete with respect to any reasonable complexity class
(see, for example, \cite{DallagoBaillot06mscs}).
In order to achieve completeness, two different extensions
of the system can be built, one adjoining basic types and constants, and 
the other adjoining second order types. In this section we 
will briefly discuss these two solutions.
\subsection{Basic Types and Constants}
Let us fix a finite set of free algebras
$\mathcal{A}=\{\A_1,\ldots,\A_n\}$. 
The constructors of $\A_i$ will
be denoted as $c_{\A_i}^1,\ldots,c_{\A_i}^{k(\A_i)}$.
The arity of constructor $c^j_{\A_i}$ will be denoted
as $\mathcal{R}_{\A_i}^j$. The algebra $\U$ of unary
integers has two constructors $c_{\U}^1,c_{\U}^2$,
where $\mathcal{R}_\U^1=1$ and $\mathcal{R}_\U^2=0$.
The languages of types, terms and values are 
extended by the the following productions
\begin{eqnarray*}
A&::=&\A_j\\
M&::=&\mathit{iter}_{\A_j}\mid\mathit{cond}_{\A_j}\mid c^i_{\A_j}\\
V&::=&\mathit{iter}_{\A_j}\underbrace{V\ldots V}_{\mbox{$k$ times}}
  \mid\mathit{cond}_{\A_j}\underbrace{V\ldots V}_{\mbox{$k$ times}}
  \mid c^i_{\A_j}\underbrace{V\ldots V}_{\mbox{$l$ times}}
\end{eqnarray*}
where $\A_j$ ranges over $\mathcal{A}$, $i$
ranges over $\{1,\ldots,k(\A_j)\}$ and 
$k$ ranges over $\{0,\ldots,k(\A_j)\}$
and $l$ ranges over  $\{0,\ldots,\mathcal{R}_{\A_j}^i\}$. 
If $t$ is a term of the free algebra $\A$ and 
$M_1,\ldots,M_{k(\A)}$ are terms, then
$t\{M_1,\ldots,M_{k(\A)}\}$ is defined by
induction on $t$: 
$(c_{\A}^it_1\ldots t_{\mathcal{R}_{\A}^i})\{M_1,$ $\ldots,M_{k(\A)}\}$
will be
$$
M_i(t_1\{M_1,\ldots,M_{k(\A)}\})\ldots(t_{\mathcal{R}_{\A}^i}\{M_1,\ldots,M_{k(\A)}\}).
$$
The new constants receive the following types
in any context:
\begin{eqnarray*}
\mathit{iter}_\A&:&\A\linear
     !(\underbrace{A\linear\ldots\linear A}_
       {\mbox{\emph{$\mathcal{R}_\A^1$ times}}}\linear A)
     \linear\ldots\linear
     !(\underbrace{A\linear\ldots\linear A}_
       {\mbox{\emph{$\mathcal{R}_\A^{k(\A)}$ times}}}\linear A)\linear !A\\
\mathit{cond}_\A&:&\A\linear
     (\underbrace{\A\linear\ldots\linear \A}_
       {\mbox{\emph{$\mathcal{R}_\A^1$ times}}}\linear A)
     \linear\ldots\linear
     (\underbrace{\A\linear\ldots\linear \A}_
       {\mbox{\emph{$\mathcal{R}_\A^{k(\A)}$ times}}}\linear A)\linear A\\
\mathit{c}_\A^i&:&\underbrace{\A\linear\ldots\linear \A}_
       {\mbox{\emph{$\mathcal{R}_\A^i$ times}}}\linear\A
\end{eqnarray*}
New (call-by-value) reduction rules are the following ones:
\begin{eqnarray*}
\mathit{iter}_\A t V_1\ldots V_{k(\A)}&\rightarrow_v & t\{V_1\ldots V_{k(\A)}\}\\
\mathit{cond}_\A c_\A^i(t_1\ldots t_{\mathcal{R}_\A^i}) V_1\ldots V_{k(\A)}
&\rightarrow_v & V_i t_1\ldots t_{\mathcal{R}_\A^i}
\end{eqnarray*}
It is easy to check that Lemma~\ref{lemma:substitution2} holds
 in the presence of the new constants. Moreover:
\begin{prop}
Every typable closed normal form is a value.
\end{prop}
\begin{proof}
By induction on the structure of a normal form $M$:
\begin{enumerate}[$\bullet$]
  \item
  A variable is not closed.
  \item
  If $M$ is an abstraction, then it is a value by definition.
  \item
  If $M$ is a constant, then it is a value by definition.
  \item
  If $M$ is an application $NL$, then $\vdash N:A\linear B$
  and $\vdash L:A$. By induction, $N$ and $L$ are both values and, as a consequence,
  $N$ cannot be a variable nor an abstraction. So, $N$ must
  be obtained from one of the new productions for values; let us
  distinguish some cases:
  \begin{enumerate}[$-$]
     \item
     If $N\equiv\mathit{iter}_{\A_j}V_1\ldots V_k$ with
     $k<k(\A_j)$, then $M$ is a value itself.
     \item
     If $N\equiv\mathit{iter}_{\A_j}V_1\ldots V_{k(\A_j)}$, 
     then $M$ is a redex, because $V_1$ is a closed value
     with type $\A_j$.
     \item
     If $N\equiv\mathit{cond}_{\A_j}V_1\ldots V_k$ with
     $k<k(\A_j)$, then $M$ is a value itself.
     \item
     If $N\equiv\mathit{cond}_{\A_j}V_1\ldots V_{k(\A_j)}$, 
     then $M$ is a redex, because $V_1$ is a closed value
     with type $\A_j$. 
     \item
     If $N\equiv c^i_{\A_j}V_1\ldots V_k$, then $k<\mathcal{R}_{\A_j}^i$
     because $N$ has an arrow type. As a consequence, $M$
     is a value.
  \end{enumerate}
\end{enumerate}
This concludes the proof.
\end{proof}
We can prove the following theorem:
\begin{thm}\label{theo:completeness}
There is a finite set of free algebras $\mathcal{A}$ including the
algebra $\U$ of unary integers such that for every 
elementary function $f:\N\rightarrow\N$, there is
a term $M_f:\U\rightarrow!^k\U$ such that
$M_f \lceil u \rceil\rightarrow^*_v \lceil f(u)\rceil$
(where $\lceil n \rceil$ is the term in $\U$ corresponding
to the natural number $n$).
\end{thm}
\begin{proof}
We will show that if $f:\N\rightarrow\N$ is computable
by a Turing Machine $\mathcal{M}$ running in elementary 
time, then there is a term $M_f$ representing that same 
function. First of all, $\mathscr{A}$ will contain
a free algebra $\C$ with four constructors
$c_\C^1,c_\C^2,c_\C^3,c_\C^4$ having
arities $\mathcal{R}_\C^1=4,\mathcal{R}_\C^2=1,\mathcal{R}_\C^3=1,
\mathcal{R}_\C^4=0$. Constructors $,c_\C^2,c_\C^3,c_\C^4$
can be used to build binary strings and a configuration
will correspond to a term $c_\C^1 t_1 t_2 t_3 t_4$ where
$t_1$ represent the current state, $t_2$ represents the
current symbol, $t_3$ represents the left-hand side of
the tape and $t_4$ represents the right-hand side of the
tape. A closed term $\mathit{trans}:\C\multimap\C$ encoding
the transition function can be built using, in particular,
the new constant $\mathit{cond}_\C$. Iteration, on the
other hand, helps when writing $\mathit{init}:\U\multimap!\C$
(whose purpose is to translate a unary integer $t$ into the initial
configuration of $\mathcal{M}$ for $t$) and $\mathit{final}:\C\multimap !\U$
(which extracts a unary integer from the final configuration of $\mathit{M}$).
In this way, the so-called qualitative part of the encoding
can be done. The quantitative part, on the other hand, can be
encoded as follows. We will show there are
terms $\mathit{tower}^n:\U\multimap !^{2n}\U$ such that
$\mathit{tower}^n\lceil m\rceil\rightarrow_v^*\lceil 2_n(m) \rceil$
where $2_0(m)=m$ and $2_{n+1}(m)=2^{2_n(m)}$ for every $n\geq 0$.
We will prove the existence of such terms by induction on $n$.
$\mathit{tower}^0:\U\rightarrow \U$ is simply the identity
$\lambda x.x$. Consider now the term
$$
\mathit{exp}\equiv \lambda x.\mathit{iter}_\U x(\lambda y\lambda z.y(yz))(\lambda y.c_\U^1 y):
\U\rightarrow !!(\U\rightarrow\U)
$$
We now prove that for every $m\in\N$,
$
\mathit{exp}\lceil m\rceil\rightarrow^*_v V_m
$
where $V_m$ is a value such that $V_m\lceil p\rceil\rightarrow_v^* \lceil 2^m +p\rceil$
for every $p\in\N$. We go by induction on $m$. If $m=0$, then
$$
\mathit{exp}\lceil m\rceil\rightarrow^*_v (\lambda x.c_\U^1 x)
$$
and $(\lambda x.c_\U^1 x)\lceil p\rceil\rightarrow_v^* \lceil 1+p\rceil\equiv\lceil 2^m+p\rceil$.
If $m>0$, then
$$
\mathit{exp}\lceil m\rceil\rightarrow^*_v (\lambda x.\lambda y.x(xy))V_{m-1}\rightarrow_v
\lambda y.V_{m-1}(V_{m-1}y)
$$
and
$$
\lambda y.V_{m-1}(V_{m-1}y)\lceil p\rceil\rightarrow^*_v V_{m-1}\lceil 2^{m-1}+p\rceil
\rightarrow^*_v\lceil 2^{m-1}+2^{m-1}+p\rceil\equiv\lceil 2^m+p\rceil
$$
$\mathit{tower}^n$ is 
$$
\lambda x.(\lambda y.\mathit{tower}^{n-1}y)((\lambda z.z\lceil 0\rceil)(\mathit{exp}\;x))
$$
Finally, we need terms $\mathit{coerc}^n:\U\rightarrow !^n\U$
such that $\mathit{coerc}^n\lceil m\rceil\rightarrow_v^*\lceil m \rceil$.
$\mathit{coerc}^0$ is simply the identity, while $\mathit{coerc}^{n}$
is $\lambda x.\mathit{iter}_\U x c_\U^1 (\lambda x.c_\U^2x)$ for every $n\geq 1$.
We can suppose there is $d$ such that $\mathcal{M}$ performs at most $2_d(n)$ steps 
processing any input of length $n$. The term $M_f$ encoding $f$ will then be:
$$
\lambda x.(\lambda y.\mathit{final}\;y)((\lambda z.\lambda v.\mathit{iter}_\U\;
z\;\mathit{trans}\;(\mathit{init}\;v))(\mathit{coerc}^{2d}\;x)(\mathit{tower}^d\;x))
$$
This concludes the proof.
\end{proof}
For the extended system a principal type property can be proved, extending the algorithm 
defined in Table \ref{table-def-type-inf} in order to take into account the new constants. 
Clearly, the system can further be extended with other constants without losing
its nice properties, provided Lemma~\ref{lemma:substitution2} is satisfied.

\subsection{Restricted Second Order Quantification}
If we had the full power of second-order quantification in \ETAS,
we would easily found a counterexample to Substitution Lemma and,
as a consequence, to Subject Reduction. Consider the following
type derivation:
$$
\infer
{x:\forall\alpha.\alpha\mid\emptyset\mid\emptyset\vdash x:!\beta}
{
  \infer{x:\forall\alpha.\alpha\mid\emptyset\mid\emptyset\vdash x:\forall\alpha.\alpha}{}
}
$$
This shows we would be able to give type $!\beta$ to the variable $x$,
but without using any instance of rule $!$. This undermines any
hope to prove subject reduction in presence of types like $\forall\alpha.\alpha$.
The same holds when we have types in the form $\forall\alpha.!A$.
Restricting second order quantification to arrow types (and, recursively,
to second-order types) allows us to preserve all results we proved in 
sections~\ref{sect:ETAS} and~\ref{sect:BOUNDS}.

Formally, a subclass of formulae can be defined by the following two productions
\begin{eqnarray*}
%A&::=& \alpha \mid\;\forall\alpha.S \\
S&::=&A\linear A\mid\;\forall\alpha.S
\end{eqnarray*}
and the following rules are added to the type system:
$$
\begin{array}{rcl}
\infer[I_\forall]{\Gamma\mid\Delta\mid\Theta\vdash M:\forall\alpha.S}{
  \Gamma\mid\Delta\mid\Theta\vdash M: S & 
  \alpha\notin\mathit{FV}(\Gamma)\cup\mathit{FV}(\Delta)\cup\mathit{FV}(\Theta)}
&\hspace{4mm}&
\infer[E_\forall]{\Gamma\mid\Delta\mid\Theta\vdash M:S\{A/\alpha\}}{
  \Gamma\mid\Delta\mid\Theta\vdash M:\forall\alpha.S}
\end{array}
$$
As can be easily checked, Theorem~\ref{theo:SR} and Theorem~\ref{theo:BOUNDS} still hold.\\
Type inference in presence of second order is at least problematic~\cite{Wells94lics}.
%If we would adjoin the full second order, the system would lose the decidability of the type inference.
%\commS{ Citare lo studente di Pedicini che ha dimostrato che aggiungere il secondo ordine alle logiche leggere 
%perde la decidibilita'. Mi sembra Pulcini: lo ha raccontato a Torino }\\
We conjecture that, even if second order quantification is the restricted one described here,
decidability of the type inference is lost.
% forse vale la decidibilita' per la decorazione con le tecniche di 
% Atassi, Baillot, Terui, ma occorre controllare

%%%%%%%%%%%%%%%%%%%%%%%%
\section{Extensions to other Logics}\label{sect:otherlightlogics}
%%%%%%%%%%%%%%%%%%
We believe the approach described in this paper to
be applicable to other logics besides Elementary Affine
Logic. Two examples are Light Affine Logic~\cite{Asperti98lics}
and Soft Affine Logic~\cite{Baillot04fossacs}.
Light Affine Logic needs an additional modality, denoted with $\S$. 
So, there will be two modal rules:
$$
\infer[!]
{\Gamma_2\mid !\Gamma_1,!\Delta_1,\Delta_2\mid\Theta_2\vdash M:!A}
{\Gamma_1\mid\Delta_1\mid\emptyset\vdash M:A & |\Gamma_1|+|\Delta_1|\leq 1}
$$
$$
\infer[\S]
{\S\Gamma_1,\S\Delta_1,\Gamma_4\mid !\Gamma_2,!\Delta_2,!\Theta_1,\Delta_4\mid\S\Theta_2,\S\Delta_3,\Theta_4\vdash M:\S A}
{\Gamma_1,\Gamma_2\mid\Delta_1,\Delta_2\Delta_3\mid\Theta_1\Theta_2\vdash M:A}
$$
Notice that we do not need an additional context for the new paragraph modality, since
contraction on formulae like $\S A$ is not allowed.

Soft Affine Logic is even simpler than elementary affine logic: there is
just one modality and the context is split into just two sub-contexts. The
$!$ rule becomes:
$$
\infer[!]
{!\Gamma_1,\Gamma_2\mid !\Delta_1,\Delta_2\vdash M:!A}
{\Gamma_1\mid\Delta_1\vdash M:A}
$$
However, the contraction in \SAL{} is deeply different from the one in \EAL{} and \LAL{}. In particular the formula $!A\lin!A\otimes!A$ is not provable anymore and is ``replaced'' by
\begin{displaymath}
!A\linear \underbrace{A\otimes\cdots\otimes A}_{\mbox{$n$ times}}.
\end{displaymath}
In the type system the axiom above is reflected by the following two rules:
\begin{displaymath}
\infer[\epsilon_1]
{\Gamma \mid\Delta,x: !A\vdash M: B}
{\Gamma\mid\Delta,x: A\vdash M:B}\qquad
\infer[\epsilon_2]
{\Gamma,x: !A\mid\Delta\vdash M: B}
{\Gamma\mid x:A,\Delta\vdash M:B}
\end{displaymath} 
Notice that the analogous of Shifting Lemma
(Lemma~\ref{lemma:lintopark} of Section~\ref{sect:ETAS}) stating that
every formula in left context can be shifted to the right one holds
in this case too.
%%%%%%%%%%%%%%%%
\section{Conclusions}\label{sect:conclusions}
%%%%%%%%%%%%%%%%%%%%%%%%
We presented a type system for the call-by-value lambda-calculus,
called~\ETAS, designed from Elementary Affine Logic and enjoying
subject reduction and elementary time normalization. Inference of
principal types can be done in polynomial time thanks to the fact that
the type system is {\em almost} syntax directed. We believe the
approach to be extendible to other systems besides \EAL, in particular
to Light Affine Logic and Soft Affine Logic (as sketched in
Section~\ref{sect:otherlightlogics}). Moreover, we show that adding
constants for iteration makes the system (extensionally) complete for
elementary time, without altering its good properties. We briefly
discuss also the problem of extending the system with second order.
%%% Local Variables: 
%%% mode: latex
%%% TeX-master: "llcbv"
%%% End: 

%%%%%%%%%%%%%%%%%%%%%%
\bibliographystyle{plain}
\bibliography{llcbv}

\begin{thebibliography}{10}

\bibitem{Asperti98lics}
Andrea Asperti.
\newblock Light affine logic.
\newblock In {\em Proceedings of the 13th IEEE Symposium on Logic in Computer
  Science}, pages 300--308, 1998.

\bibitem{Asperti02tocl}
Andrea Asperti and Luca Roversi.
\newblock Intuitionistic light affine logic.
\newblock {\em ACM Transactions on Computational Logic}, 3(1):137--175, 2002.

\bibitem{Atassi06}
Vincent Atassi, Patrick Baillot, and Kazushige Terui.
\newblock Verification of {P}time reducibility for system {F} terms via dual
  light affine logic.
\newblock {\em Logical Methods in Computer Science}, 3(4), 2007.

\bibitem{Baillot02ifip}
Patrick Baillot.
\newblock Checking polynomial time complexity with types.
\newblock In {\em Proceedings of the 2nd IFIP International Conference on
  Theoretical Computer Science}, pages 370--382, 2002.

\bibitem{Baillot04fossacs}
Patrick Baillot and Virgile Mogbil.
\newblock Soft lambda-calculus: a language for polynomial time computation.
\newblock In {\em Proceedings of the 7th International Conference on
  Foundations of Software Science and Computational Structures}, pages 27--41,
  2004.

\bibitem{Baillot04lics}
Patrick Baillot and Kazushige Terui.
\newblock Light types for polynomial time computation in lambda-calculus.
\newblock In {\em Proceedings of the 19th IEEE Syposium on Logic in Computer
  Science}, pages 266--275, 2004.

\bibitem{Baillot05tlca}
Patrick Baillot and Kazushige Terui.
\newblock A feasible algorithm for typing in elementary affine logic.
\newblock In {\em Proceedings of the 8th International Conference on Typed
  Lambda-Calculi and Applications}, pages 55--70, 2005.

\bibitem{Bohm85tcs}
Corrado B{\"o}hm and Alessandro Berarducci.
\newblock Automatic synthesis of typed lambda-programs on term algebras.
\newblock {\em Theoretical Computer Science}, 39:135--154, 1985.

\bibitem{Cop-DLag-Ron:EALCBV-05}
Paolo Coppola, Ugo Dal~Lago, and Simona Ronchi Della~Rocca.
\newblock Elementary affine logic and the call by value lambda calculus.
\newblock In Pawel Urzyczyn, editor, {\em TLCA'05}, volume 3461 of {\em Lecture
  Notes in Computer Science}, pages 131--145. Springer, 2005.

\bibitem{Coppola01tlca}
Paolo Coppola and Simone Martini.
\newblock Typing lambda terms in elementary logic with linear constraints.
\newblock In {\em Proceedings of the 6th International Conference on Typed
  Lambda-Calculi and Applications}, pages 76--90, 2001.

\bibitem{Coppola03tlca}
Paolo Coppola and Simona {Ronchi della Rocca}.
\newblock Principal typing in elementary affine logic.
\newblock In {\em Proceedings of the 6th International Conference on Typed
  Lambda-Calculi and Applications}, pages 90--104, 2003.

\bibitem{DallagoBaillot06mscs}
Ugo Dal~Lago and Patrick Baillot.
\newblock Light affine logic, uniform encodings and polynomial time.
\newblock {\em Mathematical Structures in Computer Science}, 16(4):713--733,
  2006.

\bibitem{Girard98ic}
Jean-Yves Girard.
\newblock Light linear logic.
\newblock {\em Information and Computation}, 143(2):175--204, 1998.

\bibitem{Lafont02SLL}
Yves Lafont.
\newblock Soft linear logic and polynomial time.
\newblock {\em Theoretical Computer Science}, 318(1-2):163--180, 2004.

\bibitem{Terui02phd}
Kazushige Terui.
\newblock {\em Light logic and polynomial time computation}.
\newblock PhD thesis, Keio University, 2002.

\bibitem{Terui01lics}
Kazushige Terui.
\newblock Light affine lambda calculus and polytime strong normalization.
\newblock {\em Archive for Mathematical Logic}, 46(3-4):253--280, 2007.

\bibitem{Wells94lics}
Joe Wells.
\newblock Typability and type checking in the second-order lambda -calculus are
  equivalent and undecidable.
\newblock In {\em Proceedings of the 9th IEEE Symposium on Logic in Computer
  Science}, pages 176--185, 1994.

\end{thebibliography}
%%%%%%%%%%%%%%%%%%%%%%%%

\end{document}